\documentclass{article}
\usepackage{setspace} 
\usepackage{amsmath,amsthm,amssymb}
\usepackage[a4paper]{geometry}
\usepackage{graphicx}
\usepackage{natbib}
\usepackage{float}
\usepackage{subcaption}
\usepackage{url}
\usepackage[dvipsnames]{xcolor}
\usepackage{multirow}
\usepackage[ruled]{algorithm2e}
\usepackage{algpseudocode}
\usepackage[hidelinks,colorlinks=true,linkcolor=blue,citecolor=blue,urlcolor=blue]{hyperref} 
\usepackage{placeins}
\usepackage{afterpage}
\SetArgSty{textup}

\makeatletter
\newcommand*\if@single[3]{%
  \setbox0\hbox{${\mathaccent"0362{#1}}^H$}%
  \setbox2\hbox{${\mathaccent"0362{\kern0pt#1}}^H$}%
  \ifdim\ht0=\ht2 #3\else #2\fi
  }
\newcommand*\rel@kern[1]{\kern#1\dimexpr\macc@kerna}
\newcommand*\widebar[1]{\@ifnextchar^{{\wide@bar{#1}{0}}}{\wide@bar{#1}{1}}}
\newcommand*\wide@bar[2]{\if@single{#1}{\wide@bar@{#1}{#2}{1}}{\wide@bar@{#1}{#2}{2}}}
\newcommand*\wide@bar@[3]{%
  \begingroup
  \def\mathaccent##1##2{%
    \if#32 \let\macc@nucleus\first@char \fi
    \setbox\z@\hbox{$\macc@style{\macc@nucleus}_{}$}%
    \setbox\tw@\hbox{$\macc@style{\macc@nucleus}{}_{}$}%
    \dimen@\wd\tw@
    \advance\dimen@-\wd\z@
    \divide\dimen@ 3
    \@tempdima\wd\tw@
    \advance\@tempdima-\scriptspace
    \divide\@tempdima 10
    \advance\dimen@-\@tempdima
    \ifdim\dimen@>\z@ \dimen@0pt\fi
    \rel@kern{0.6}\kern-\dimen@
    \if#31
      \overline{\rel@kern{-0.6}\kern\dimen@\macc@nucleus\rel@kern{0.4}\kern\dimen@}%
      \advance\dimen@0.4\dimexpr\macc@kerna
      \let\final@kern#2%
      \ifdim\dimen@<\z@ \let\final@kern1\fi
      \if\final@kern1 \kern-\dimen@\fi
    \else
      \overline{\rel@kern{-0.6}\kern\dimen@#1}%
    \fi
  }%
  \macc@depth\@ne
  \let\math@bgroup\@empty \let\math@egroup\macc@set@skewchar
  \mathsurround\z@ \frozen@everymath{\mathgroup\macc@group\relax}%
  \macc@set@skewchar\relax
  \let\mathaccentV\macc@nested@a
  \if#31
    \macc@nested@a\relax111{#1}%
  \else
    \def\gobble@till@marker##1\endmarker{}%
    \futurelet\first@char\gobble@till@marker#1\endmarker
    \ifcat\noexpand\first@char A\else
      \def\first@char{}%
    \fi
    \macc@nested@a\relax111{\first@char}%
  \fi
  \endgroup
}
\makeatother

\theoremstyle{plain}
\newtheorem{prop}{Proposition}

\newtheorem{lemm}[prop]{Lemma}
\newtheorem{theo}[prop]{Theorem}
\theoremstyle{definition}
\newtheorem{exam}{Example}
\newtheorem{defi}{Definition}
\newtheorem{assu}{Assumption}
\newtheorem{rema}{Remark}

\newcommand{\lamh}{\hat{\lambda}}
\newcommand{\convp}{\overset{p}{\to}}
\newcommand{\tV}{\widetilde{V}}
\newcommand{\EE}[2][]{\mathbb{E}_{#1}\left[#2\right]}
\newcommand{\argmax}{\operatorname{argmax}}
\newcommand{\argmin}{\operatorname{argmin}}
\newcommand{\PP}[2][]{\mathbb{P}_{#1}\left[#2\right]}
\newcommand{\simiid}{\,{\buildrel \text{iid} \over \sim\,}}

\newcommand\indep{\protect\mathpalette{\protect\independenT}{\perp}}
\def\independenT#1#2{\mathrel{\rlap{$#1#2$}\mkern2mu{#1#2}}}

\begin{document}

\title{Qini Curves for Multi-Armed Treatment Rules}

\author{
Erik Sverdrup\\
\footnotesize{\texttt{Monash University}}
\and 
Han Wu\\
\footnotesize{\texttt{Two Sigma}}
\and
Susan Athey\\
\footnotesize{\texttt{Stanford University}}
\and
Stefan Wager\\
\footnotesize{\texttt{Stanford University}}
}

\maketitle

\begin{abstract}
Qini curves have emerged as an attractive and popular approach for evaluating the benefit of data-driven targeting rules for treatment allocation. We propose a generalization of the Qini curve to multiple costly treatment arms that quantifies the value of optimally selecting among both units and treatment arms at different budget levels. We develop an efficient algorithm for computing these curves and propose bootstrap-based confidence intervals that are exact in large samples for any point on the curve. These confidence intervals can be used to conduct hypothesis tests comparing the value of treatment targeting using an optimal combination of arms with using just a subset of arms, or with a non-targeting assignment rule ignoring covariates, at different budget levels. We demonstrate the statistical performance in a simulation experiment and an application to treatment targeting for election turnout. 
\end{abstract}

\section{Introduction}
The Qini curve, initially proposed in the marketing literature \citep{radcliffe2007using}, plots the average policy effect of treating the units most responsive to the treatment as we vary the budget. We can then quantify the value of treatment targeting via cost-benefit evaluations undertaken at a series of distinct budget levels. The Qini curve has been adopted in a variety of practical applications to evaluate the empirical performance of treatment targeting rules \citep{albert2022commerce,belbahri2021qini, diemert2018large, uplift_review, rzepakowski2012decision, zhao2017uplift, zhao2019uplift}.

The theoretical properties of Qini-like metrics under a binary treatment, and extensions to area under the curve summaries, have recently received attention in the statistics literature by a number of authors, including \citet{imai2021experimental, imai2022statistical, sun2021treatment}, and \citet{yadlowsky2021evaluating}. These approaches consider the problem of targeting the assignment of a (possibly costly) binary intervention. In this paper, we explore the extension to scenarios where there are multiple treatment arms, and where the benefits and costs of assignment may vary across units. For example, a low-cost drug may be beneficial for a certain group of people, but a high-cost drug may be even more beneficial for a subset of these. Analyzing this setting through separate Qini curves for the two available arms can conceal important efficiency trade-offs. For a specific budget, the optimal policy may entail assigning different drugs to different people; a less expensive drug for one group and a costlier drug for another. Determining the optimal treatment assignment policy that maps individual characteristics to one of several treatment arms involves solving a constrained optimization problem.

We develop a theoretical and statistical framework to extend the Qini curve to the case where we have many mutually exclusive and costly treatments. We show that the Qini curve extended to multiple arms retains the desirable ratio-based interpretation of the Qini for a single treatment arm, where it is the \textit{incremental efficiency} of each arm that determines the optimal allocation. This means that it is not necessary to denominate treatment effects and costs on the same scale. An additional unit of budget is allocated to an arm and a set of targeted participants (defined by their characteristics) if the ratio of the benefits to the set of participants relative to the cost is greater than the corresponding ratio for any other arm and set of participants.    

\begin{figure}[t]
\centering
    \includegraphics[width=0.6\textwidth]{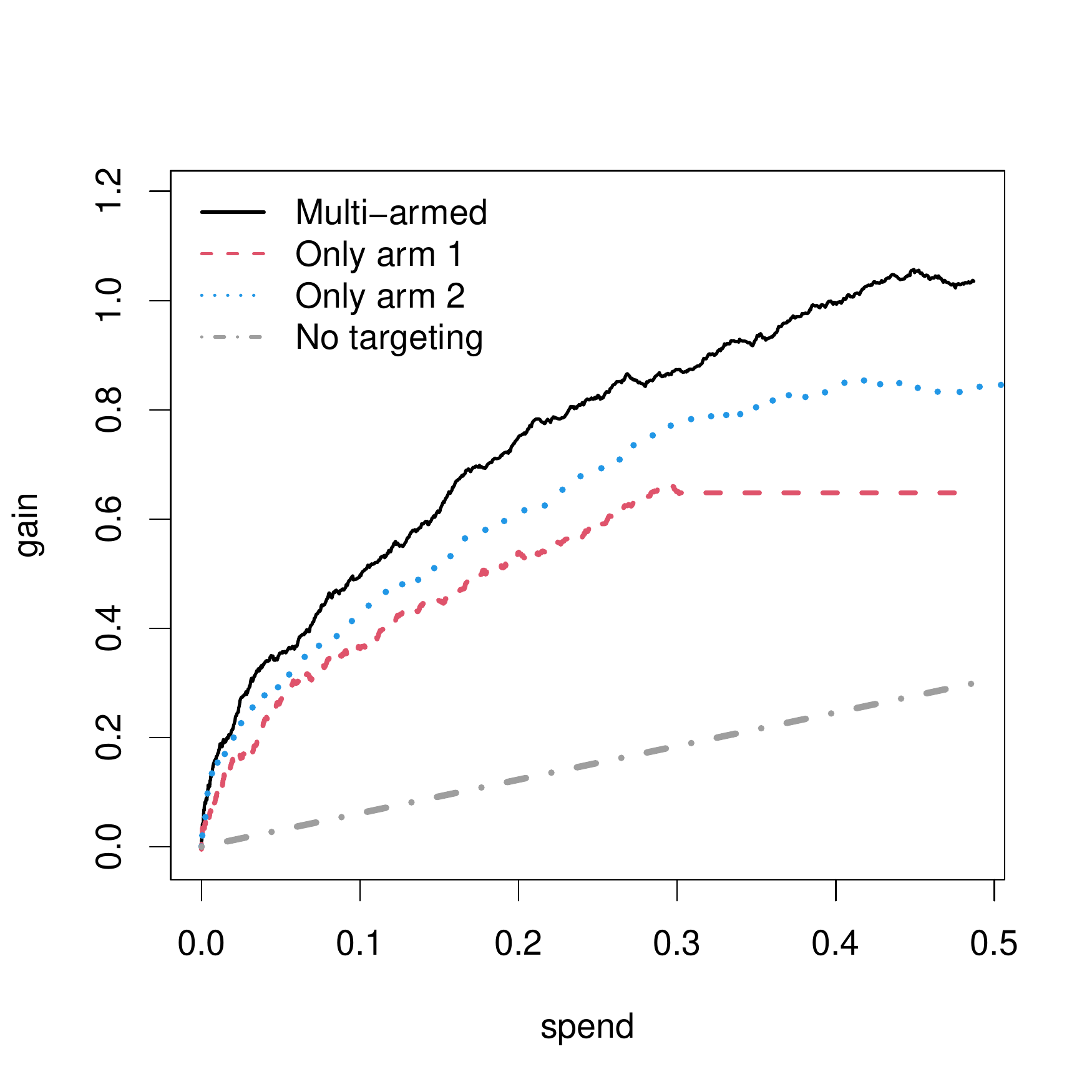}
\caption{Qini curves for single-armed treatment policies and a Qini curve for a multi-armed policy, using synthetic data described in Section \ref{sec:simulation}. The no-targeting line represents the value of randomly allocating units to arm 1 and traces out an average treatment effect. Arm 2 has a negative estimated average treatment effect and so its no-targeting line is omitted from the plot.}
\label{fig:fig1}
\end{figure}

To gain an intuition for the generalization of the Qini to multiple arms, recall that the Qini curve is a metric for evaluating a \emph{treatment rule} implied by a policy. With a single treatment arm, where for simplicity the cost of assignment is the same for each unit, the optimal policy is to allocate treatment in decreasing order of the conditional average treatment effect. Given estimates of these treatment effects, the Qini curve plots the value of assigning treatment to individuals as prioritized by their estimated treatment effects. Figure \ref{fig:fig1} shows Qini curves in a setting with two treatment arms. For example, if we can only use arm 1 and have a total budget of 0.2, then we can achieve a gain of 0.54; whereas if we can only use arm 2 the same budget yields an estimated gain of 0.60. Note that, once we pass a spend level of 0.3, the arm-1 Qini curve plateaus---this is because, once we have reached this spend level using arm 1, we are already giving treatment to all units believed to benefit from it, and so cannot achieve further gains via increased spending.

The single-armed curves in Figure \ref{fig:fig1} are straightforward to compute, as the underlying policies induce a priority rule that involves sorting units in order of the estimated conditional average treatment effect. Computing the optimal allocation for a multi-armed policy is more complicated, as it involves solving a constrained cost-benefit problem across many arms. We show that, even though the underlying multi-armed policies are more complicated, they still yield a tractable treatment rule that can be evaluated with Qini curves. The Qini curve for a multi-armed policy in Figure \ref{fig:fig1} highlights that since different arms can be better for different groups, targeting enables the arms to be assigned accounting for the cost-benefit analysis appropriate for distinct subgroups. For example, with a budget of 0.2, we can now achieve a gain of 0.75, which is better than what we could get with either arm alone.

Incorporating additional arms beyond two improves (i.e., raises) the Qini curve for two reasons. First, even in the absence of targeting, expanding the budget leads to greater use of arms that on average are less efficient (lower cost-benefit ratio) but are relatively beneficial. Second, targeting allows the identification of subgroups who particularly benefit from arms that might perform poorly on average, and thus not be prioritized in the absence of targeting.

We characterize the optimal multi-armed policy, showing that when expanding the budget, the optimal assignment selects units to receive more effective treatments according to where the incremental cost-benefit ratio is highest. We further show how, for given characteristics of a unit, the optimal policy can be characterized by a set of budget thresholds where the unit's assignment changes to a more beneficial but less efficient arm. We propose an efficient algorithm for estimating the solution path of the multi-armed policies that underlie the Qini curve by leveraging a connection to the linear multiple-choice knapsack problem. 

Our main theoretical result quantifies uncertainty for points on the Qini curve via a central limit theorem for the estimated multi-armed policy values. The result takes estimates of conditional average treatment effects and costs as given but accounts for the uncertainty from approximating the optimal allocation for each level of budget, and from estimating the policy value for that allocation. The central limit theorem can be used to, for example, quantify the value of employing more arms when targeting treatment.

An open-source software implementation, \texttt{maq} for \texttt{R} and \texttt{Python}, is available at \href{https://github.com/grf-labs/maq}{\textcolor{PineGreen}{github.com/grf-labs/maq}}.

\section{The Solution Path for Optimal Multi-Armed Treatment Assignment}\label{sec:optV}
To characterize the optimal multi-armed treatment allocation, we operate under the potential outcomes framework \citep{imbens2015causal}. We assume that we observe independent and identically distributed samples $(X_i, W_i, Y_i, C_i) \simiid P$ for $i=1,\ldots,n$, where $X_i \in \mathcal{X}$ denotes pre-treatment covariates, $W_i \in \{0, 1,\ldots, K\}$ denotes the treatment assignment ($W_i=0$ is the control group), $Y_i \in \mathbb{R}$ denotes the observed outcome, and $C_i \in \mathbb{R}$ denotes the incurred cost of assigning the unit the given treatment. We posit the potential outcomes $\{Y_i(0),\ldots, Y_i(K)\}$, $\{C_i(0), \ldots, C_i(K)\}$ and we assume $Y_i = Y_i(W_i)$  and $C_i = C_i(W_i)$ (SUTVA).

For the mutually exclusive treatment arms $k = 1,\ldots, K$, let $\tau(X_i)$ and $C(X_i)$ denote the vectors of conditional average treatment effects and cost contrasts, i.e. the $k$-th elements are:
\begin{align}
    \tau_k(x) &= \EE{Y_i(k) - Y_i(0) \mid X_i = x}, \\
    C_k(x) &= \EE{C_i(k) - C_i(0) \mid X_i = x}.
\end{align}
In our intended application, the role of arm $k=0$ is to provide us with an option to hold back treatment without incurring any costs. To accommodate this setup we need to assume that the cost of assigning arm this arm is 0 and the remaining arms have non-zero costs,
\begin{assu} \label{assu:cost}
$C_i(0) = 0$ and $C_i(k) \geq C_i(0)$ almost surely and $\EE{C_i(k)-C_i(0) \mid X_i = x} > 0$ for all $k = 1,\ldots,K$.
\end{assu}

Our goal is to gain insight into how much there is to gain from treatment targeting if treatment is assigned optimally. To do so, denote a policy by $\pi: \mathcal{X} \rightarrow \mathbb{R}^{K}$, a mapping from covariate $X_i$ to a treatment assignment. The policy $\pi(X_i)$ is a $K$-dimensional vector where the $k$-th element is equal to 1 if arm $k$ is assigned, and zero otherwise.\footnote{Fractional assignments between 0 and 1 are admissible and can be interpreted as probabilistic assignment between arms.} The associated value of this treatment assignment policy is the expected value:\footnote{In the policy learning literature it is sometimes common to define the value of a policy via potential outcome means \citep{athey2021policy}. Had we instead encoded $\pi$ to be a non-fractional policy taking values in the set $\{0, \, 1, \ldots, K\}$ then an equivalent formulation of the gain \eqref{eq:gaindef} would be $V(\pi) := \EE{Y(\pi(X_i)) - Y_i(0)}$.}
\begin{defi}
The expected gain (policy value) of a treatment assignment policy is the expected value it achieves in comparison to assigning each unit the control arm,
\begin{equation}\label{eq:gaindef}
       V(\pi) = \EE{\langle \pi(X_i), \tau(X_i) \rangle},
\end{equation}
where the notation $\langle a, b\rangle$ denotes an inner product between vectors $a$ and $b$.
\end{defi}
The cost of this policy is $\Psi(\pi) = \EE{\langle \pi(X_i), C(X_i)\rangle}$. Throughout the manuscript, we assume the cost function $C(\cdot)$ is known (i.e., the cost structure is a modeling decision). The optimal policy is the one that, for a given budget level, maximizes the expected gain while incurring costs less than or equal to the budget in expectation. Given a budget $B$, the optimal unrestricted policy $\pi^*_B$ that only depends on $X_i$ solves the following stochastic optimization problem:
\begin{equation}\label{eq:qiniopt}
    \pi^*_B = \argmax \{V(\pi): \Psi(\pi) \leq B\}.
\end{equation}
The asterisk is to emphasize that this policy has access to the population oracle $\tau(\cdot)$ function.
In the case of only a single treatment arm ($K=1$), but where each unit's cost may be different, \eqref{eq:qiniopt} is an instance of the fractional knapsack problem \citep{dantzig1957discrete} and the optimal policy induces an appealing treatment rule allocating treatment to units in decreasing order of the cost-benefit ratio $\EE{Y_i(1) - Y_i(0) \mid X_i=x} / \EE{C_i(1) - C_i(0) \mid X_i=x}$ until the budget runs out \citep{luedtke2016optimal, sun2021treatment}. The treatment allocation in this induced ranking constitutes the solution path over varying budget levels.

The multi-armed case ($K > 1$) is more complicated, as \eqref{eq:qiniopt} then belongs to the class of multiple-choice knapsack problems \citep{sinha1979multiple}, a type of optimization problem that involves filling a knapsack up to a capacity by selecting at most one item from a set of classes, where each item has an associated ``profit'' and ``weight''. In our formulation, the class is a unit and the item is a treatment arm with the profits and weights corresponding to the conditional average treatment effect and cost of the particular arm. The knapsack capacity is the budget constraint. Allowing for fractional treatment allocation reduces this problem to a linear program with $nK$ choice variables. Using the transformation principles presented in \citet{zemel1980linear}, it is possible to recast this into inducing a similar treatment priority rule, but where the priority is based on ``incremental'' cost-benefit ratios.

\subsection{Characterizing the Optimal Polices}\label{sec:characterization}
The idea behind characterization via incremental cost-benefit ratios is to recast the problem of choosing between both units and treatment arms into thresholding a suitable priority rule that captures both which unit and which arm is optimal to assign at a given budget level. For any given unit $i$, the only treatment arms that will be active in the optimal solution are the ones that lie on the convex hull of the cost-reward plane \citep[Proposition 2]{sinha1979multiple}. For any $x \in \mathcal{X}$, define the convex hull formed by the points $(C_k(x), \tau_k(x)), \, k = 0,\ldots,K$ to be a set of $m_x$ points with the ordering $k_1(x), \ldots, k_{m_x}(x)$ such that 
\begin{equation*}
    0 = C_{k_1(x)}(x) < \cdots  < C_{k_{m_x}(x)}(x)
\end{equation*}
\begin{equation*}
    0 = \tau_{k_{1}(x)}(x) < \cdots  < \tau_{k_{m_x}(x)}(x)
\end{equation*}
\begin{equation*}
    \rho_{k_{1}(x)}(x) > \cdots \rho_{k_{m_x}(x)}(x) > 0
\end{equation*}
where we define the incremental cost-benefit ratio as
\begin{equation} \label{eq:rho_def}
    \rho_{k_{j}(x)}(x) := \frac{\tau_{k_{j}(x)}(x)-\tau_{k_{j-1}(x)}(x)}{C_{k_{j}(x)}(x) - C_{k_{j-1}(x)}(x)}
\end{equation}
and we let $\rho_0(x) = \infty$ and $\rho_{k}(x) = -\infty$ if $k \notin \{k_1(x),\ldots,k_{m_x}(x)\}$. 

\begin{figure}[t]
\centering
    \includegraphics[width=0.7\textwidth]{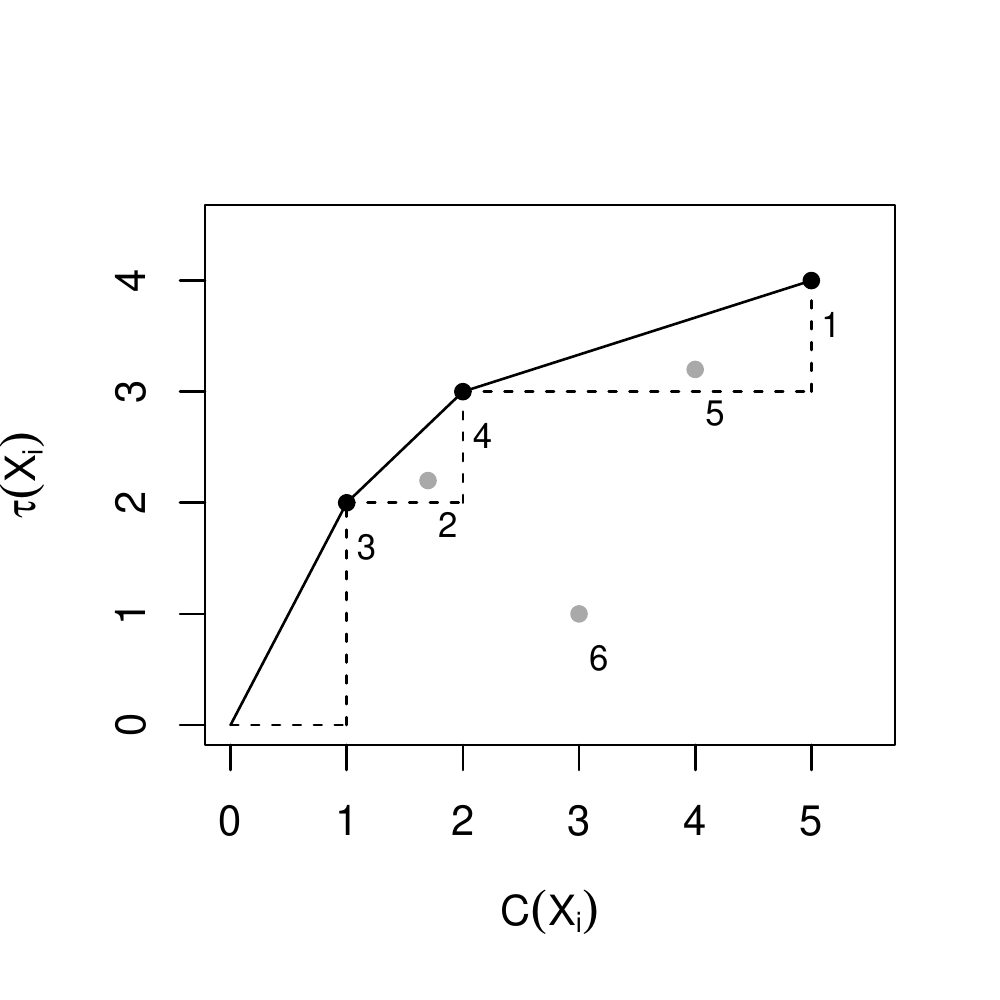}
\caption{Treatment arms (numbers listed beneath circles) on the ($C$, $\tau$) plane. The points $(0,0), (1,2), (2,3), (5,4)$ are on the convex hull and $k_1(x) = 0, k_2(x) = 3, k_3(x) = 4, k_4(x) = 1$. And by our definition, $\rho_0(x) = \infty, \rho_1(x) = \frac{1}{3}, \rho_4(x) = 1, \rho_3(x) = 2, \rho_k(x) = -\infty$ for $k = 2,5,6$.}
\label{fig:lp_dominated}
\end{figure}

Figure \ref{fig:lp_dominated} illustrates the case of optimally assigning treatment for a single unit $i$. If we have an available budget of 1, it would be optimal to assign arm 3 to the $i$-th unit. If we increase the available budget to 2, then we have two choices: upgrade to either arm 2 or 4. Since arm 2 lies outside the convex hull, it is strictly sub-optimal to assign this arm, and the optimal assignment is arm 4. For the optimal policy, we are faced with a distribution of convex hulls, one for each unit, and have to decide whether to assign a new unit a treatment or upgrade an existing unit to a costlier arm. In Theorem \ref{theo:opt_policy}, we show that the insight from \citet{zemel1980linear} carries over to the stochastic setting: What matters in each of these convex hulls are the slopes of the tangent lines between arms, i.e., the incremental cost-benefit ratio \eqref{eq:rho_def}. For a given budget level, when choosing between selecting an arm for unit $i$ or $j$, the (unit, arm) with the largest tangent slope is optimal. The following theorem formalizes this intuition by characterizing the optimal stochastic policy at a given budget level $B$, in terms of thresholding of the distribution of incremental cost-benefit ratios.

\begin{defi}
\label{defi:multi_thresh}
Given cost-benefit ratios $\rho$ as in \eqref{eq:rho_def}, a threshold $\lambda \geq 0$ and an
interpolation value $c \in [0, \, 1)$, we define a multi-armed thresholding rule as
\begin{equation}
T_{k_j(x)}(x; \, \rho, \, \lambda, \, c) := \begin{cases}
1 & \text{if }  \ \rho_{k_j(x)}(x) > \lambda > \rho_{k_{j+1}(x)}(x), \\
c & \text{if }  \ \rho_{k_j(x)}(x) = \lambda, \\
1-c & \text{if } \ \rho_{k_{j+1}(x)}(x) = \lambda, \\
0 & \text{otherwise}.
\end{cases}
\end{equation}
\end{defi}

\begin{theo} \label{theo:opt_policy}
Under Assumption \ref{assu:cost}, there exists an optimal (stochastic) policy $\pi_B^{*}$
that takes the form of a multi-armed thresholding rule following Definition \ref{defi:multi_thresh}, i.e.,
there exist constants $\lambda_B \geq 0$ and $c_B \in [0, \, 1)$ such that
\begin{equation} \label{eq:pi_opt}
\pi^*_{B}(x) = T(x; \, \rho, \, \lambda_B, \, c_B).
\end{equation}
\end{theo}

For generic distributions where $X$ has continuous support, $\mathbb{P}[\rho_{k_j(x)}(x) = \lambda_B] = 0$ for all $\lambda_B > 0$, and so the optimal policy will almost surely be integer-valued.

\section{The Qini Curve for Multi-Armed Policies}\label{sec:qinicurve}

Section \ref{sec:optV} provides a characterization that maps a budget $B$, a cost function $C(\cdot)$ and the population function $\tau(\cdot)$, to an optimal policy $\pi_B^*(X_i)$. Given an independent and identically distributed random sample from this population, we can obtain, through various estimation methods, estimates $\hat \tau(\cdot)$ of the function $\tau(\cdot)$ (see Section \ref{sec:simulation} for an example). We refer to the sample used to obtain these estimates as the \emph{training sample}. These estimates induce a policy:
\begin{defi} \label{def:induced_policy}
Let $\hat \tau(\cdot)$ be the estimates of the conditional average treatment effect function obtained on a training sample. The induced policy $\pi_B$ is the policy that solves
\begin{equation}
    \pi_B = \argmax_{\pi} \left\{ \EE{\langle \pi(X_i), \hat \tau(X_i) \rangle}: \EE{\langle \pi(X_i), C(X_i) \rangle} \leq B \right\}, \label{eq:population_induced_policy}
\end{equation}
i.e., we are solving \eqref{eq:qiniopt} but replacing the population quantity $\tau(\cdot)$ with the estimated function $\hat \tau(\cdot)$.
\end{defi}
As a metric to evaluate treatment allocation according to an induced policy, we define the Qini curve:
\begin{defi} \label{def:qinicurve}
Given a family of policies $\pi_B$ indexed by $(\hat \tau, C$), the Qini curve is the curve that plots the function $Q(B) = V(\pi_B), \, B \in (0, B_{max}]$.
\end{defi}

The challenge now is, once we have a \emph{test sample} of independent and identically distributed random sample from the population, how do we form estimates of $Q(B)$? To keep concepts clear we define the empirical induced policy on the test set:

\begin{defi} \label{def:induced_test_policy}
Consider $n$ independently and identically distributed test samples from the population. Let $\hat \tau(\cdot)$ be the estimate of the conditional average treatment effect function obtained from a training sample. The test set empirical induced policy $\hat \pi_B$ is the policy that solves
\begin{equation}
    \hat \pi_B = \argmax_{\pi} \left\{\frac{1}{n} \sum_{i=1}^{n} \langle \pi(X_i), \hat \tau(X_i) \rangle: \frac{1}{n} \sum_{i=1}^{n} \langle \pi(X_i), C(X_i) \rangle \leq B \right\}, \label{eq:test_induced_policy}
\end{equation}
i.e., we are solving \eqref{eq:population_induced_policy} over an empirical test sample indexed by units $i=1\ldots n$.
\end{defi}

In order to form an estimate of $Q(B)$ on a test sample, there are three subsequent challenges we need to address: how to handle the budget constraint, how to efficiently express $\hat \pi_B$, and finally, how to estimate the policy value of $\pi_B$. The first issue, we address by satisfying the budget in expectation on the test set as in Definition \ref{def:induced_test_policy}.

\paragraph{Expressing $\hat \pi_B$ on the test set.} The optimization problem in \eqref{eq:test_induced_policy} has a linear program formulation that takes the following form,
\begin{equation}
    \begin{aligned}\label{eqn:naiveLP}
        \max_{\pi} \quad & \frac{1}{n}\sum_{i=1}^{n} \sum_{k=1}^{K} \pi_k(X_i) \hat \tau_k(X_i)\\ 
        \textrm{s.t.} \quad & \frac{1}{n}\sum_{i=1}^{n} \sum_{k=1}^{K} \pi_k(X_i) C_k(X_i) \leq B, \\
        &\sum_{k=1}^{K} \pi_k(X_i) \leq 1, \, i = 1 \ldots n, \\
        &\pi_k(X_i) \geq 0,~ k=1\ldots K,~ i = 1 \ldots n.
    \end{aligned}
\end{equation}
The direct approach of solving \eqref{eqn:naiveLP} via generic LP-solvers is computationally infeasible as this would involve computing a large collection of linear programs with $nK$ choice variables, one for each budget constraint $B \in (0, B_{max}]$.

A more feasible approach is to instead directly compute the path of solutions $\{\hat \pi_B\}_{B \rightarrow 0}^{B_{max}}$ via an algorithm tailored to the structure \eqref{eqn:naiveLP} embeds. To this end, the characterization of the optimal policy in Theorem \ref{theo:opt_policy} as a thresholding rule of incremental cost-benefit ratios $\rho$ is promising as it suggests the problem can be reduced to a single-dimensional fractional knapsack problem (with some additional bookkeeping). This is exactly the approach taken by \citet{zemel1980linear}, to solve \eqref{eqn:naiveLP} via sorting the incremental cost-benefit ratios \citep[Chapter~11]{kellerer2004multidimensional}.\footnote{Faster algorithms for the LP-relaxation of the multiple-choice knapsack problem exits, \citep{dyer1984n, zemel1984n} derive linear-time solutions for a fixed budget level, but these do not readily generalize to a path algorithm.} Figure \ref{fig:figure_lambda} illustrates how $\rho$ determines a solution. The vertical axis shows the incremental cost-benefit ratios for each unit's arm on the convex hull (with units indexed by the horizontal axis). A solution to \eqref{eqn:naiveLP} is given by a particular threshold $\lambda_B$ on the vertical axis and determines the optimal allocation through a planar separation of unit-arm pairs.

A limitation of the algorithm in \citet{zemel1980linear} is that it solves \eqref{eqn:naiveLP} at only a single budget level $B$, as determined by a single planar separation. In order to adapt this algorithm to deliver a path of solutions over budget levels, we can make use of a priority queue ordered by decreasing $\rho$ that acts as a construction that keeps track of which (unit, arm) enters the active set of the solution path, as we lower $\lambda_B$ in Figure \ref{fig:figure_lambda} in accordance with the budget $B$ we are tracing out.

\begin{figure}[t]
    \centering
    \begin{subfigure}[b]{0.475\textwidth}
        \centering
        \includegraphics[width=\textwidth]{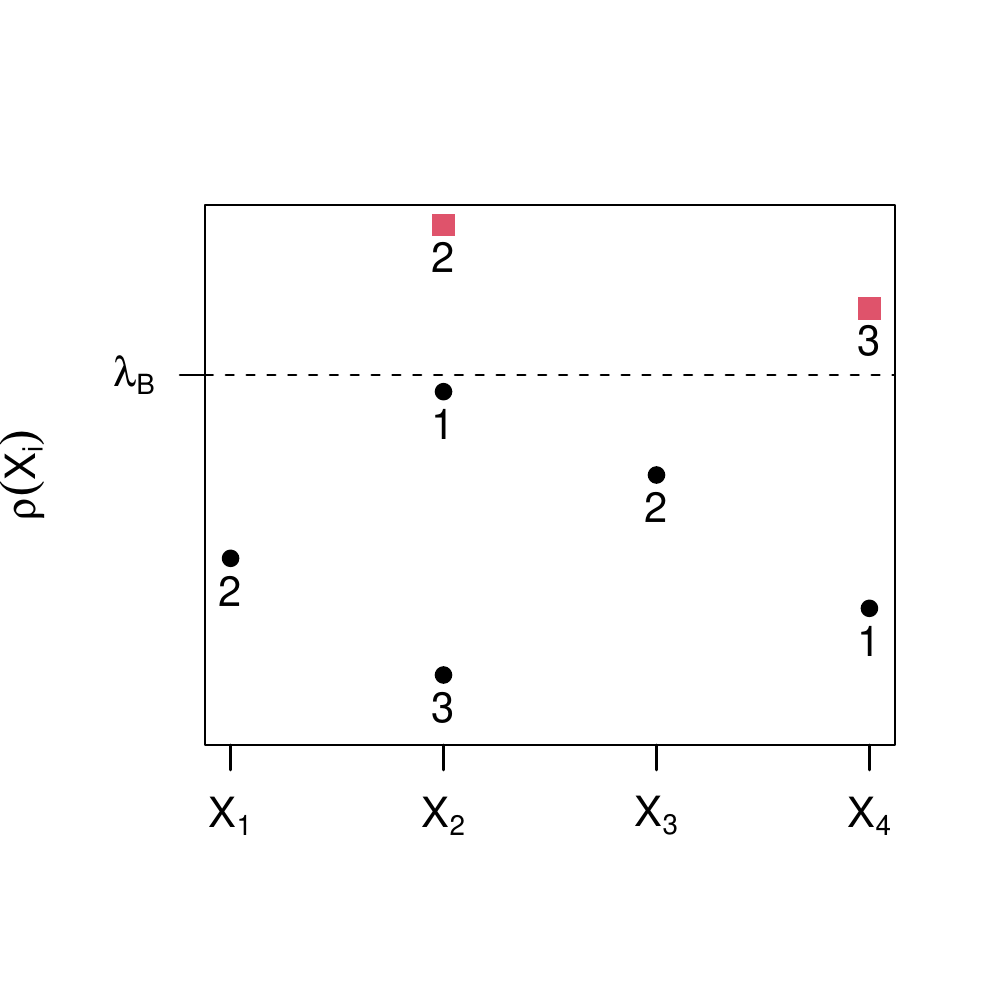}
        \caption[]
        {{At the current $\lambda_B$ the optimal solution is to assign arm 2 to unit $X_2$ and arm 3 to $X_4$ (squares).}} 
        \label{fig:figure_lambda_a}
    \end{subfigure}
    \hfill
    \begin{subfigure}[b]{0.475\textwidth}  
        \centering 
        \includegraphics[width=\textwidth]{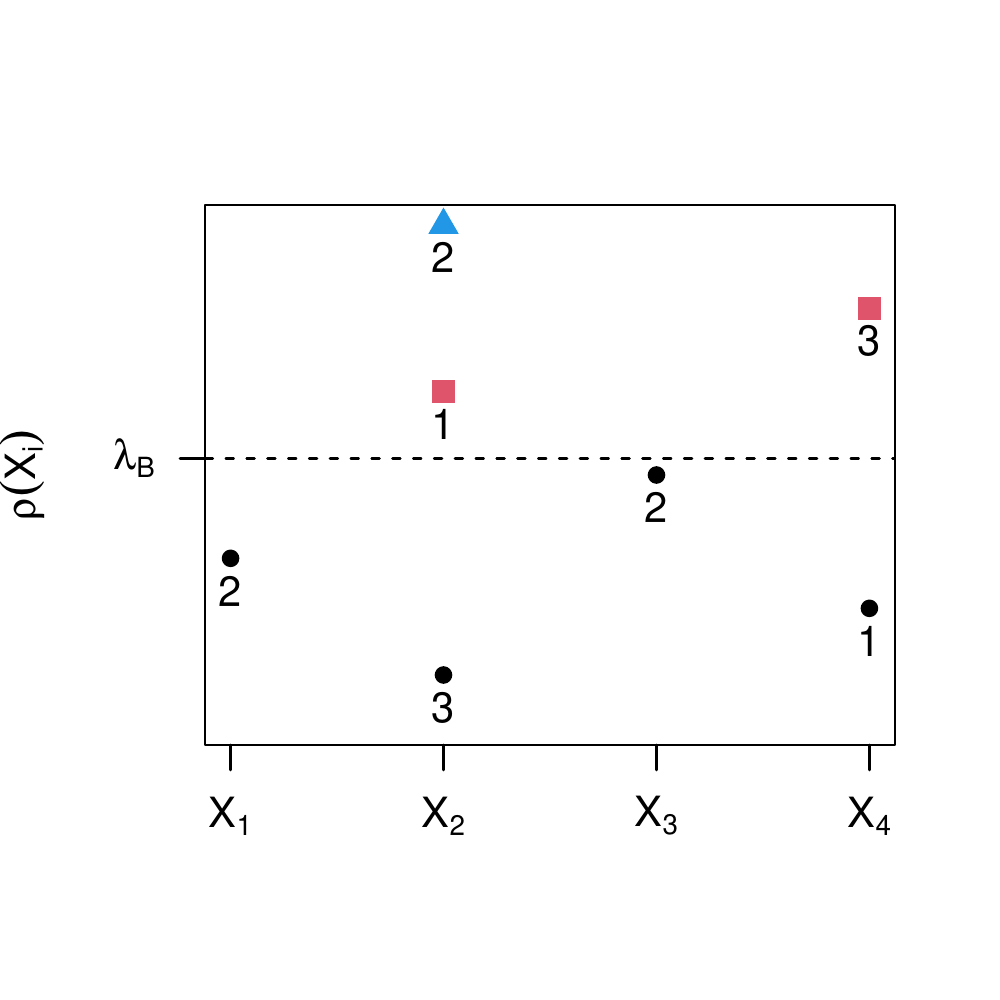}
        \caption[]
        {{At a higher budget level arm 1 is assigned unit $X_2$.\\}}    
        \label{fig:figure_lambda_b}
    \end{subfigure}
    \caption{Illustration of how for 4 example units $X_i$ the incremental cost-benefit ratios $\rho$ trace out a solution path over increasing budget levels by decreasing $\lambda_B$. Squares are active arms in the current solution, circles are inactive arms waiting to be assigned, and triangles are previous arm assignments at past budget constraints. Recall that $\rho_{k_j(x)}(x)$ measures the incremental cost-benefit ratio of upgrading a unit with covariates $x$ to arm $j$ from the previous arm along its convex hull \eqref{eq:rho_def}.}\label{fig:figure_lambda}
\end{figure}

\paragraph{Estimating the value of $\pi_B$.} Now that we have a promising strategy to obtain the estimated path $\{\hat \pi_B\}_{B \rightarrow 0}^{B_{max}}$, how do we estimate its value? We show in Section \ref{sec:inference} that the approximation error of the empirical optimization in \eqref{eq:test_induced_policy} is asymptotically linear with zero means. This means we can leverage standard policy evaluation arguments for this component. Thus, with a suitable construction $\widehat \Gamma_i$ that satisfies $\EE{\widehat \Gamma_i \mid X_i} \approx \tau(X_i)$ policy evaluation arguments motivates forming an estimate of $Q(B)$ with the plug-in construction\footnote{Note that we are using the empirical induced policy $\hat\pi_B$ (Definition \ref{def:induced_test_policy}) to estimate the value of the population induced policy $\pi_B$ (Definition \ref{def:induced_policy}). Theorem \ref{theo:asymp_linear_reward} verifies the validity of this approach.}

\begin{equation}\label{eq:qiniplugin}
\widehat Q(B) = \widehat V(\pi_B) = \frac{1}{n} \sum_{i=1}^{n}\langle \hat \pi_B(X_i), \widehat \Gamma_i \rangle,
\end{equation}
where $\widehat \Gamma_i$ could be obtained with, in the case of known treatment randomization probabilities (as in a randomized controlled trial), inverse propensity weighting \citep{horvitz1952generalization}:
\begin{equation}\label{eq:IPW}
    \widehat \Gamma_{i,k} = \frac{\mathbf{1}(W_i=k)Y_i}{\PP{W_i=k}} - \frac{\mathbf{1}(W_i=0)Y_i}{\PP{W_i=0}}.
\end{equation}
If the treatment randomization probabilities are unknown (as would be the case in an observational setting under unconfoundedness), then a more robust alternative to forming \eqref{eq:IPW} via plug-in estimates of the propensity score, is to rely on augmented inverse propensity weighting \citep{robins1994estimation}. This approach relies on nuisance estimates in the form of propensity scores $e_k(x) = \PP{W_i = k \mid X_i = x}$ and conditional response surfaces $\mu_{k}(x) = \EE{Y_i \mid X_i = x, W_i = k}$. In order to construct this score, these components need to be estimated on the test set data.

To ensure that these estimated components are independent of the outcome for each unit, a popular approach is to employ cross-fitting \citep{dml, schick1986asymptotically} where the $i$-th unit's estimate is obtained without using that unit for estimation. The multi-armed score then takes the following form \citep{robins1994estimation, zhou_multi_action}, with $k$-th entry equal to
\begin{equation}\label{eq:drAIPW}
\widehat \Gamma_{i, k} =
{\hat \mu^{-q(i)}}_k(X_i) - {\hat \mu^{-q(i)}}_0(X_i) +
 \left( \frac{\mathbf{1}(W_i=k)}{\hat e^{-q(i)}_{k}(X_i)} -
\frac{\mathbf{1}(W_i=0)}{\hat e^{-q(i)}_{0}(X_i)} \right) (Y_i - {\hat \mu^{-q(i)}}_{W_i}(X_i)),
\end{equation}
where the super script $-q(i)$ denotes fitting using the data excluding the fold $X_i$ belongs to. This approach for evaluation can yield an efficiency gain over inverse-propensity weighting (see \citet[Section 2.3]{yadlowsky2021evaluating} for a discussion).

\paragraph{Computing the solution path and values.} With all the pieces needed to estimate $Q(B)$ in place, Algorithm \ref{alg:algopath_full} outlines pseudo-code for all the components needed to compute the Qini curve for a multi-armed policy, starting with estimating conditional average treatment effects on a training set. With these, a chosen cost function, and suitable evaluation scores in place, Algorithm \ref{alg:algopath} formalizes the intuition behind Figure \ref{fig:figure_lambda} with pseudo-code for computing the induced multi-armed policy and value up to some maximum budget level $B_{max}$.

\begin{algorithm}[p]
    \caption{Estimate the Qini for a multi-armed policy. With all test set estimates constructed the solution path is computed with \texttt{ComputePath} described in Algorithm \ref{alg:algopath}.}
    \label{alg:algopath_full}
    \SetKwInOut{Input}{Input}
    \SetKwInOut{Output}{Output}
    \Input{
    Set of training samples $\mathcal{S}_{train} = \{X_i, W_i, Y_i\}_{i=1}^{n_{train}}$,\\
    ~Set of test samples $\mathcal{S}_{test} = \{X_i, W_i, Y_i\}_{i=1}^{n}$,\\
    ~Cost function $C(\cdot)$,\\
    ~The maximum budget per unit, $B_{max}$.}
    \Output{~The estimated Qini curve $\widehat Q(B)$.}

    {\textbf{1)} Estimate $\hat \tau(\cdot)$ on the training set $\mathcal{S}_{train}$.}\\
    {~$\hat \tau(\cdot) \leftarrow$ \texttt{EstimateCATE}$\left(\mathcal{S}_{train}\right)$}\\

    {\textbf{2)} Predict $\hat \tau$ on the test set $\mathcal{S}_{test}$.}\\
    {~$\hat \tau(X_i) \leftarrow$ \texttt{PredictCATE}$\left(X_i \in \mathcal{S}_{test}\right)$}\\

    {\textbf{3)} Form evaluation scores $\widehat \Gamma$ on the test set $\mathcal{S}_{test}$.}\\
    ~\If{randomization probabilities are known}{
    {Construct the scores according to inverse-propensity weighting \eqref{eq:IPW}.}\\
    {~$\widehat \Gamma_{i,k} = \dfrac{\mathbf{1}(W_i=k)Y_i}{\PP{W_i=k}} - \dfrac{\mathbf{1}(W_i=0)Y_i}{\PP{W_i=0}}$}
    }
    \Else{
    {Form cross-fit nuisance estimates.}\\
    {~$\hat \mu^{-q(i)}(X_i) \leftarrow$ \texttt{EstimateCrossFitResponses}$\left(\mathcal{S}_{test}\right)$}\\
    {~$\hat e^{-q(i)}(X_i) \leftarrow$ \texttt{EstimateCrossFitPropensities}$\left(\mathcal{S}_{test}\right)$}\\
    {Construct the scores according to augmented inverse-propensity weighting \eqref{eq:drAIPW}.}\\
    {~$\widehat \Gamma_{i,k} =
    {\hat \mu^{-q(i)}}_k(X_i) - {\hat \mu^{-q(i)}}_0(X_i) +
     \left( \dfrac{\mathbf{1}(W_i=k)}{\hat e^{-q(i)}_{k}(X_i)} -
    \dfrac{\mathbf{1}(W_i=0)}{\hat e^{-q(i)}_{0}(X_i)} \right) (Y_i - {\hat \mu^{-q(i)}}_{W_i}(X_i))$}
    }

    {\textbf{4)} Estimate the induced policy path and value on $\mathcal{S}_{test}$ using Algorithm \ref{alg:algopath}.}\\
    {~~$\{\widehat V(\pi_B),~ \widehat \Psi(\pi_B)\}_{B\rightarrow 0}^{B_{max}} \leftarrow$ \texttt{ComputePath}$\left(\hat \tau(X_i),~ C(X_i),~ \widehat \Gamma_i,~ B_{max},~ i=1,\ldots,n\right)$}\\

    \Return $\{\widehat V(\pi_B),~ \widehat \Psi(\pi_B)\}_{B \rightarrow 0}^{B_{max}}$
\end{algorithm}

\begin{algorithm}[p]
    \caption{(\texttt{ComputePath}) Compute the multi-armed policy solution path. Time complexity: $O(nK \log {nK})$, where $n$ is the number of test samples and $K$ is the number of treatment arms.}
    \label{alg:algopath}
    \SetKwInOut{Input}{Input}
    \SetKwInOut{Output}{Output}
    \Input{
    Test set treatment effect estimates $\{\hat \tau(X_i)\}_{i=1}^{n}$,\\
    ~Test set costs $\{C(X_i)\}_{i=1}^{n}$,\\
    ~Test set evaluation scores $\{\widehat \Gamma_i\}_{i=1}^{n}$,\\
    ~Maximum budget per unit, $B_{max}$.}
    \Output{~A vector of gain estimates over increasing spend levels up to $B_{max}$.} 
    \For{all test samples $x$} {
        {Compute the set of arms on the convex hull.}\\
        {$\{k_1(x), \ldots, k_{m_x}(x)\} \leftarrow$ \texttt{ComputeConvexHull}$\left(\hat \tau(X_i),~ C(X_i),~ X_i = x\right)$} \algorithmiccomment{See \ref{sec:appendix_cvx_hull}.}\\
    }
    
    {$gain \leftarrow \varnothing$}\algorithmiccomment{Initialize gain $\widehat V(\cdot)$ and spend  $\widehat \Psi(\cdot)$ to empty vectors.}\\
    {$spend \leftarrow \varnothing$}\\
    {$pqueue  \leftarrow \text{PriorityQueue()} $} 
    
    \For{all test samples $x$} {
        {$\hat \rho(x) \leftarrow \dfrac{\hat \tau_{k_{1}(x)}(x)}{C_{k_{1}(x)}(x)}$ }\algorithmiccomment{Compute the efficiency of initial arm on convex hull.}\\
        {$pqueue.\text{add}$($(x, k_1(x))$ with priority $\hat \rho(x)$)} \algorithmiccomment{Enqueue each unit's initial arm.}\\
    }
    \While{current $spend$ $<$ $B_{max}$ and $pqueue$.size() $> 0$}{
        {$(x, k_j(x))  \leftarrow   pqueue.\text{pop}()$}\algorithmiccomment{Retrieve unit and arm on top of queue.}\\
        \If{already assigned an arm to unit $x$}{
            {Subtract previous arm's cost and gain from current $spend$ and $gain$.}
        }
        {Allocate arm $k_j(x)$ to unit $x$, record gain and pay for it.}\\
        {$gain$.append$(\text{current}~ gain + \widehat \Gamma_{k_{j}(x)}(x)/n)$ }\\
        {$spend$.append$(\text{current}~ spend + C_{k_{j}(x)}(x)/n)$}\\ 

        \If{current $spend$ $> B_{max}$}{
            {Perform fractional adjustment for unit $x$.} \algorithmiccomment{Given by the constant $c_B$ in \eqref{theo:opt_policy}}.\\
            {\textbf{break}}
        }
        
        \If{there remain arms on convex hull for unit $x$}{
            {$k_{j+1}(x) \leftarrow$ next arm on the convex hull}\\
            {$\hat \rho(x) = \dfrac{\hat \tau_{k_{j+1}(x)}(x)-\hat \tau_{k_{j}(x)}(x)}{C_{k_{j+1}(x)}(x) - C_{k_{j}(x)}(x)}$}\algorithmiccomment{Compute the incremental efficiency.}\\
            {$pqueue.\text{add}$($(x, k_{j+1}(x))$ with priority $\hat \rho(x)$)} \algorithmiccomment{Enqueue the next arm.}\\
        }
    }
    \Return $\{gain, \, spend\}$
\end{algorithm}

After a reduction to convex hulls, Algorithm \ref{alg:algopath} starts by adding each unit's first arm on the convex hull to a priority queue ordered by decreasing estimates $\hat \rho$ of the cost-benefit ratios. The first unit assigned is the unit on top of this queue (top square in Figure \ref{fig:figure_lambda_a}). If this unit has remaining arms on its convex hull (i.e., there are arms below the unit's initial allocation in Figure \ref{fig:figure_lambda}), then this subsequent arm is added to the queue with priority equal to its incremental cost-benefit ratio. The subsequent assignments might either be upgrades, in which case we move to a costlier arm lower on the vertical plane, or a new unit allocation. The exact sequence of upgrade-or-allocate-new-unit decisions is dictated by the priority queue order $\hat \rho$. The time complexity of this algorithm is log-linear in $nK$.\footnote{To give an impression of the practical performance of using this as an evaluation metric, for a sample size of one million, and five treatment arms, our open-source implementation computes the full solution path in around 1.5 seconds on a standard laptop.}

Depending on the value of $B_{max}$, the treatment allocation for the last unit to be assigned might not be integer-valued. By Theorem \ref{theo:opt_policy} there are two such cases. The first case is if the $i$-th unit has previously not been assigned an arm, and there is not sufficient budget left to allocate the first arm on the convex hull. The second case is if the $i$-th unit has previously been assigned an arm, but there is not sufficient budget left to upgrade the unit to the next arm on the convex hull. In these cases, we may think of assigning the $i$-th unit an arm with a certain probability, as given by the fractional allocation $c_B$. In our intended setup, treatment is assigned to a large number $N$ of units matching the covariate profile of $X_i$, a fractional solution would simply mean that, in the second case, we assign one arm to $c_B N$ units, and the other arm to the remaining $(1-c_B)N$ units.

Finally, while Algorithm \ref{alg:algopath} does not explicitly construct and return the vectors $\hat \pi_B(X_i)$, these are implicitly given by the sequence of (unit, arm) allocations and can be efficiently constructed ex-post, which is the approach taken in the accompanying software.

\afterpage{\FloatBarrier}

\subsection{A Central Limit Theorem for the Qini Curve}\label{sec:inference}
In order to employ the Qini curve for decision-making, we need to form the uncertainty estimate of $\widehat V(\pi_B)$, a point on the curve (we consider functionals such as area under the curves as an interesting extension for future work). In this section, we provide an asymptotic linearity theorem for the policy value estimate, which enables confidence intervals and hypothesis tests via resampling-based methods \citep{chung2013exact,yadlowsky2021evaluating}.

We first make some standard identifying assumptions on the population, 
\begin{assu} [Overlap] \label{assu:overlap}
    There exists $\eta$ such that $e_k(x) > \eta$ for all $x$ and $k$, where $e_k(x) = \PP{W_i = k \mid X_i = x}$. 
\end{assu}
\begin{assu} [Unconfoundedness] \label{assu:unconf}
    $Y_i(0), \ldots, Y_i(K) \indep W_i \mid X_i$.
\end{assu} 

Recall from Theorem \ref{theo:opt_policy} that the policy $\pi_B$ in Definition \ref{def:induced_policy} can be expressed via the thresholding function $T(\cdot \, ; \, \hat \rho, \, \lambda_B, \, c_B)$, where $\hat \rho$ are estimates of the cost-benefit ratios using the estimated function $\hat \tau(\cdot)$, and $C(\cdot)$. The argument $\lambda_B$ in this function represents the value of the threshold of $\hat \rho$ that maximizes the gain under the given treatment and cost rules while satisfying the budget in population.  The empirical policy $\hat \pi_B$ we obtain in Definition \ref{def:induced_test_policy} is a function $T(\cdot \, ; \, \hat \rho, \, \hat \lambda_B, \, \hat c_B)$, i.e., the estimated policy is determined by the empirical threshold $\hat \lambda_B$ (the argument $\hat \rho$ is the same as we are fixing the $\hat \tau(\cdot)$ and $C(\cdot)$ functions that give rise to the treatment rule under consideration).

Our goal is to verify that the empirical threshold $\hat \lambda_B$ converges to the population threshold $\lambda_B$, and that the plug-in construction \eqref{eq:qiniplugin} serves as a valid estimate of $V(\pi_B)$. The convenience of representing $\pi_B$ via the function $T$ is that it yields expressions of our objects of interest as functions of the threshold $\lambda_B$. To this end, we make the following assumptions on the gain and cost functions:

\begin{assu} [Continuity] \label{assu:contin}
    The gain
    $$
    V(T(\cdot \, ; \, \hat \rho, \, \lambda, \, c)) = \EE{\langle T(X_i; \, \hat \rho, \, \lambda, \, c), \tau(X_i) \rangle},
    $$
    and mean cost
    $$
    \Psi(T(\cdot \, ; \, \hat \rho, \, \lambda, \, c)) = \EE{\langle T(X_i; \, \hat \rho, \, \lambda, \, c), C(X_i)\rangle},
    $$
    are continuously differentiable with respect to $\lambda$ (denoted by $V'(\cdot)$ and $\Psi'(\cdot)$), and the set of covariates $X$ having the incremental cost-benefit ratio $\hat \rho$ exactly $\lambda$ has measure 0.
\end{assu}

With suitable assumptions in place, recall from Section \ref{sec:qinicurve} that there are two components needed to form an estimate of a point on the Qini curve, an empirical induced policy and an evaluation score $\widehat \Gamma_i$. As outlined in the previous paragraph, we have an exact expression for the first component, via an estimated threshold $\hat \lambda_B$. This yields a representation of the estimated policy value \eqref{eq:qiniplugin} via 
\begin{equation}
  \widehat Q(B) = \frac{1}{n} \sum_{i=1}^{n}\langle T(X_i ; \, \hat \rho, \, \hat \lambda_B, \, \hat c_B), \widehat \Gamma_i \rangle.
\end{equation}
Our goal is to quantify the uncertainty in estimating $V(T(\cdot \, ; \, \hat \rho, \, \lambda_B, \, c_B)) = \EE{\langle T(X_i; \, \hat \rho, \, \lambda_B, \, c_B), \tau(X_i) \rangle}$ through the sampling variability of this plug-in estimate. This construction has two levels of approximation: using an estimated threshold $\hat \lambda_B$, arising from solving for the empirical induced policy via empirical optimization on the test set, and using an estimated score $\widehat{\Gamma}_i$ constructed on the test set in order to construct the estimate $\widehat V(T(\cdot \, ; \, \hat \rho, \, \hat \lambda_B, \, \hat c_B))$. 

If we were using a fixed deterministic $\lambda_B$, the asymptotic property of $\widehat V(T(\cdot \, ; \, \hat \rho, \, \lambda_B, \, c_B))$ follows from classical doubly robust arguments.
The main challenge for proving central limit theorems for Qini curves is that the budget-satisfying threshold $\lamh_B$ (and thus the actual policy deployed on the test set) is random. So results in, e.g., \citet{athey2021policy}, on test-set evaluation of fixed policies, are not directly applicable. Our idea is to argue asymptotic linearity by building on results of \citet{yadlowsky2021evaluating} who prove a central limit theorem for Qini curves in the two-arm setting, by first proving asymptotic linearity of the threshold $\lamh_B$. Then, we combine with the standard doubly robust argument to prove that $\widehat V(T(\cdot \, ; \, \hat \rho, \, \lambda_B, \, c_B))$ satisfies asymptotic linearity. 

To argue about $\lamh_B$, we note that we can view $\lamh_B$ as an approximate Z-estimator assuming the empirical threshold $\lamh_B$ approximately makes the cost equal to the budget $B$ on the test set. The following theorem details the argument and states a central limit theorem; the overall architecture outlining where the various estimates come from is in Algorithm \ref{alg:algopath_full}.

\begin{theo} \label{theo:asymp_linear_reward}
Under Assumption \ref{assu:cost}, \ref{assu:overlap}, \ref{assu:unconf}, \ref{assu:contin}, let $\hat \tau(\cdot)$ be any estimate of the CATE function, fit on an independent training set. $C(\cdot)$ is a chosen cost function. Suppose all potential outcomes are bounded. Let $\pi_B$ with a threshold $\lambda_B$ be the induced policy \eqref{eq:population_induced_policy} with respect to $\hat \tau(\cdot)$, i.e. $\pi_B(x) = T(x; \, \hat \rho, \, \lambda_B, \, c_B)$ solves the following equation 
\begin{equation}
    \EE{\langle T(X_i; \, \hat \rho, \, \lambda_B, \, c_B), C(X_i) \rangle} = B.
\end{equation}
Let $\hat \pi_B(x) = T(x; \, \hat \rho, \, \hat \lambda_B, \, \hat c_B)$ be the empirical policy \eqref{eq:test_induced_policy} induced by the budget constraint on a test sample of $n$ points $\{X_i, W_i, Y_i\}_{i=1}^{n}$, such that
\begin{equation} \label{eq:approxlambda}
    \frac{1}{n} \sum_{i=1}^{n} \langle T(X_i; \, \hat \rho, \, \hat \lambda_B, \, \hat c_B), C(X_{i}) \rangle - B = o_p(n^{-1/2}).
\end{equation}
Let $a_i$ be the arm assigned to unit $X_i$, and assume further that $\hat \rho_{a_i}(x_i)$ has a continuous density in a neighborhood of $\lambda_B$ for any $i$. Assume that we construct doubly robust scores $\widehat \Gamma_i$ with cross-fitting on the test set using \eqref{eq:drAIPW}, with the following assumptions on the estimates of the nuisance components $\mu$ and $e$: 
\begin{itemize}
    \item The estimates are sup-norm consistent.  
    \item The estimates satisfy the following error bounds 
\begin{equation}\label{assump:consistency_eqv}
\EE{\Big( \hat{\mu}_{k}^{-q(i)}(X_i) - \mu_{k}(X_i) \Big)^2} \cdot
\EE{\Big( \hat{e}_{k}^{-q(i)}(X_i) - e_{k}(X_i) \Big)^2} = o(1/n),~ k=0,\ldots, K.
\end{equation}
\end{itemize}
Let $\psi(x) = \langle T(x; \, \hat \rho, \, \lambda_B, \, c_B), C(x) \rangle$ - B. Then $\widehat V(T(\cdot \, ; \, \hat \rho, \, \hat \lambda_B, \, \hat c)) = \frac{1}{n} \sum_{i=1}^{n}\langle T(X_i; \, \hat \rho, \, \hat \lambda_B, \, \hat c_B), \widehat \Gamma_i \rangle$ is asymptotically linear, with the following expansion
\begin{align}
    &n^{1/2}\left(\widehat V(T(\cdot \, ; \, \hat \rho, \, \hat \lambda_B, \, \hat c_B))- V(T(\cdot \, ; \, \hat \rho, \, \lambda_B, \, c_B))\right) \nonumber \\
    &= n^{-1/2} \sum_{i=1}^n \left(\langle \pi_B(X_i), \Gamma_i\rangle -  \frac{V'(\pi_B)}{\Psi'(\pi_B)}\psi(X_i) - V(\pi_B) \right)+ o_p(1),
\end{align}
where $\Gamma_i$ is the oracle doubly robust score with $k$-th entry
\begin{equation}
\Gamma_{i,k} =
{\mu}_k(X_i) - {\mu}_0(X_i) +
 \left( \frac{\mathbf{1}(W_i=k)}{e_{k}(X_i)} -
\frac{\mathbf{1}(W_i=0)}{e_{0}(X_i)} \right) (Y_i - {\mu}_{W_i}(X_i)).
\end{equation}
\end{theo}
In Theorem \ref{theo:asymp_linear_reward} we condition on the training set used to obtain the conditional average treatment effect function, and consider the randomness on the test set used to evaluate the induced policies. Asymptotic linearity justifies bootstrap-based inference of the Qini curves, in particular, it makes half-sampling a suitable choice for resampling Algorithm \ref{alg:algopath} \citep[Lemma 4]{yadlowsky2021evaluating}. To compute one single bootstrap replicate, rerun Algorithm \ref{alg:algopath} on a random half-sample of units to obtain a path of policy value estimates, then interpolate this to the grid of spend values on the path computed for the full sample. As only half of the samples are passed to Algorithm \ref{alg:algopath}, the evaluation score $\widehat \Gamma_j$ for the $j$-th drawn unit is given a weight equal to 2. Algorithm \ref{alg:bootstrap} outlines the steps.

This framework for evaluating multi-armed treatment rules can also be applied in observational settings where Assumption \ref{assu:unconf} does not hold, by instead relying on ``proximal'' identifying assumptions \citep{tchetgen2020introduction}, where we only have access to noisy measurements of the confounding variables. In this setting, similar arguments as in \citet{pmlr-v202-sverdrup23a} can be applied to learn a function $\tau(\cdot)$, and then replace the score construction $\widehat \Gamma_i$ with the proximal doubly robust score of \citet{cui2020semiparametric}.

\begin{algorithm}[hp]
    \caption{Compute bootstrapped standard errors of the Qini curve. Estimates from Algorithm \ref{alg:algopath_full} are taken as given and Algorithm \ref{alg:algopath} is run on a subsample from $\mathcal{S}_{test}$.}
    \label{alg:bootstrap}
    \SetKwInOut{Input}{Input}
    \SetKwInOut{Output}{Output}
    \Input{
    Test set treatment effect estimates $\{\hat \tau(X_i)\}_{i=1}^{n}$ from step 2 in Algorithm \ref{alg:algopath_full},\\
    ~Test set costs $\{C(X_i)\}_{i=1}^{n}$,\\
    ~Test set evaluation scores $\{\widehat \Gamma_i\}_{i=1}^{n}$ from step 3 in Algorithm \ref{alg:algopath_full},\\
    ~Maximum budget per unit, $B_{max}$,\\
    ~Number of bootstrap replicates $R$.}
    \Output{~A vector of standard error estimates over increasing spend levels up to $B_{max}$.} 

    \For{$r=1 \ldots R$} {
        {Draw a random half from $n$ samples: $j = j_1 \ldots j_{\lfloor{n/2}}$\\}
        {Estimate the $r$-th solution path on the bootstrapped sample using Algorithm \ref{alg:algopath}.}
        {$\{\widehat V(\pi_B)^{(r)}  \}_{B\rightarrow 0}^{B_{max}} \leftarrow$ \texttt{ComputePath}$\left(\hat \tau(X_j),~ C(X_j),~ \widehat \Gamma_j,~ B_{max},~ j = j_1 \ldots j_{\lfloor{n/2}}\right)$}\\
    }
    
    {Form estimates of the standard error of the policy value for each $B$.}\\
    \For{$B \rightarrow 0 \ldots B_{max}$} {
        {$\sigma(\widehat V(\pi_B)) \leftarrow \dfrac{1}{R} \mathlarger \sum_{r=1}^{R}\left( \widehat V(\pi_B)^{(r)} - \dfrac{1}{R} \mathlarger \sum_{r=1}^{R} \widehat V(\pi_B)^{(r)} \right)^2$}
    }
    \Return $\{ \sigma(\widehat V(\pi_B)) \}_{B \rightarrow 0}^{B_{max}}$
\end{algorithm}

\section{Simulation Experiment}\label{sec:simulation}
There is a wide variety of strategies available to estimate conditional average treatment effects $\tau(X_i)$ that can be extended to the multi-armed setting, including \citet{athey2016recursive}, \citet{kennedy2020drlearner}, \citet{kunzel2019metalearners}, \citet{nie2020quasi}, and \citet{seibold2016model}.
In the empirical illustrations, we use the forest-based multi-armed treatment effect estimator of \citet*{wager2019grf}, which is based on the \emph{R-learner} framework \citep{nie2020quasi}. The method is available via the \texttt{R} package \texttt{grf} \citep{GRF, Rcore} via the function \texttt{multi\_arm\_causal\_forest}, which has built-in functionality to produce the multi-armed evaluation scores \eqref{eq:drAIPW}. This approach estimates $\tau(X_i)$ directly using the following forest-weighted loss
\begin{align*}
\hat \tau(x) = \argmin_{\tau} \left\{ \sum_{i=1}^{n}
 \alpha_i (x) \left( Y_i - \hat m^{(-i)}(X_i) - c(x) -
 \left\langle 1_{W_i} - \hat e^{(-i)}(X_i), \,  \tau(X_i)  \right\rangle
 \right)^2 \right\},
\end{align*}
where $\hat m$ are estimates of the conditional mean function marginalizing over treatment $\EE{Y_i \mid X_i = x}$, $\hat e$ are estimates of the propensity scores, and the superscript $(-i)$ indicates that the estimates for the $i$-th observation is obtained without using unit $i$ for training. The forest weights $\alpha(x)$ are adaptive nearest neighbor weights obtained by a generalized random forest \citep{wager2019grf} searching for heterogeneity in the vector-valued target parameter $\tau(X_i)$.

As a synthetic illustration, we adapt the three-armed data generating process in \citet{zhou_multi_action}, treating the first arm as a zero-cost control, with covariates $X_i$ identically and independently distributed on $[0, 1]^{10}$, and potential outcomes distributed according to
\begin{equation*}
    \EE{Y_i(w_i) \mid X_i} = (3 - w_i)\mathbf{1_0}(X_i) + (2 - 0.5|w_i - 1|)\mathbf{1_1}(X_i) + 1.5(w_i - 1)\mathbf{1_2}(X_i), 
\end{equation*}
where $\mathbf{1_0}(X_i), \mathbf{1_1}(X_i), \mathbf{1_2}(X_i)$ indicate which region a unit belongs to:
\begin{align*}
    &\mathbf{1_0}(X_i) = \mathbf{1}(X_{i5} \leq 0.6) \mathbf{1}(X_{i7} \geq 0.35), \\
    &\mathbf{1_1}(X_i) = \mathbf{1}\left(X_{i5}^2/0.6^2 + X_{i7}^2/0.35^2 < 1\right) + \mathbf{1}\left((X_{i5}-1)^2/0.4^2 + (X_{i7}-1)^2/0.35^2 < 1\right), \\
    &\mathbf{1_2}(X_i) = 1 - \mathbf{1_0}(X_i) - \mathbf{1_1}(X_i).
\end{align*}
We let the assignment probabilities for the different arms be the same,
\begin{align*}
    &\PP{W_i=0 \mid X_i} = 1/3,~ \PP{W_i=1 \mid X_i} = 1/3,~ \PP{W_i=2 \mid X_i} = 1/3.
\end{align*}
We treat the cost for the two treatment arms as known and equal to a unit's observable pre-treatment covariates, $C_i(1) = X_{i1}, C_i(2) = 2X_{i2}$. Outcomes are observed with noise $N(0, 4)$.

To study the practical inferential properties of points on the Qini curve for multiple arms, using flexible machine learning estimators, we calculate coverage of 95\% confidence intervals for $Q(B)$. We first fix a $\hat \tau(\cdot)$ function estimated on a training set with $n=10000$. We consider ten points $B = \{0.05, 0.10, 0.15, 0.20, 0.25, 0.30, 0.35, 0.4, 0.45, 0.5\}$ on the Qini curve, then on a test set with size $n=\{1000, 2000, 5000, 10000\}$ compute the policy $\hat \pi_B$, estimate doubly robust scores $\widehat \Gamma$, then calculate coverage of the estimated $Q(B)$ using bootstrapped standard errors. The results in Table \ref{tab:CI} show the mean empirical coverage of this procedure over $1000$ Monte Carlo repetitions. In Appendix \ref{appendix:sim} we repeat this exercise, but where we form the nuisance estimates needed to construct the multi-armed evaluation score \eqref{eq:drAIPW} via boosting instead of random forests.

\begin{table}[ht]
\begin{center}
\begin{tabular}{c|rrrrrrrrrr}
      \multicolumn{6}{r}{Spend ($B$)}\\
      Sample size & 0.05 & 0.1 & 0.15 & 0.2 & 0.25 & 0.3 & 0.35 & 0.4 & 0.45 & 0.5 \\
      \hline
        1000 & 0.95 & 0.95 & 0.95 & 0.95 & 0.95 & 0.95 & 0.95 & 0.95 & 0.95 & 0.94 \\
        2000 & 0.95 & 0.95 & 0.95 & 0.95 & 0.94 & 0.95 & 0.95 & 0.95 & 0.95 & 0.95 \\
        5000 & 0.95 & 0.95 & 0.96 & 0.96 & 0.95 & 0.95 & 0.95 & 0.95 & 0.95 & 0.95 \\
        10000 & 0.95 & 0.93 & 0.94 & 0.94 & 0.94 & 0.94 & 0.94 & 0.95 & 0.95 & 0.94 \\
    \end{tabular}
\caption{Coverage (\%) of the 95\% confidence intervals for $Q(B)$ at ten spend points ($B$) using the simulation setup described in Section \ref{sec:simulation}. The number of Monte Carlo repetitions is 1000. The number of bootstrap replicates for standard error estimation is 200.}
\label{tab:CI}
\end{center}
\end{table}

\section{Hypothesis Tests for Treatment Targeting Strategies}\label{sec:targeting}
Our proposed method can be used to compare different targeting rules in two ways.
A first use case is as a tool for practitioners to quantify the benefit of employing more arms as presented in, e.g., Figure \ref{fig:fig1}.
A second use case is as a tool for practitioners to quantify how much benefit there is to treatment targeting.  Addressing this use case involves constructing a suitable baseline policy $\bar \pi_B$ that ignores covariates, and then quantifying the value of treatment targeting via a difference in Qini curves. In the case of only a single treatment arm, the policy that ignores covariates is simple to compute, as it for a given budget amounts to treating some fraction of the population, and the value of this policy is simply a fraction of the average treatment effect of the arm. With more than one treatment arm, $\bar \pi_B$ needs to take into account the average treatment effect and average costs of the $K$ arms, which motivates the following definition.

\begin{defi} \label{def:pibar}
For a given budget $B$, the policy $\bar \pi_B$ which ignores covariates is the policy that solves the problem in Definition \ref{def:induced_policy} with only access to the average treatment effect estimates $\bar \tau = \EE{\hat \tau(X_i)}$ and average costs $\widebar C = \EE{C(X_i)}$. The Qini curve for this policy is the function $\widebar Q(B) = V(\bar \pi_B), \, B \in (0, B_{max}]$.
\end{defi}

Intuitively, this policy collapses all the information from the $X_i$-specific convex hulls to the single convex hull traced out by $\bar \tau$ and $\widebar C$. For a given budget, it splits the budget between two consecutive arms on the convex hull. Computing the value of this policy is straightforward by using Algorithm \ref{alg:algopath} on the single convex hull traced out by $(\bar \tau$, $\widebar C)$ and evaluating it with $\widehat \Gamma_i := 1/n \sum_i \widehat \Gamma_i$.

To shed light on how to address the two targeting use cases mentioned above, Table \ref{tab:targetingsetup} presents an overview of different policy objects that arise when there are three treatment arms available. The traditional Qini curve for a single treatment arm allows for comparisons between rows and columns in the first three rows. For example, using the simplified notation in Table \ref{tab:targetingsetup}, $Q_1 - Q_2$ is the value of targeting with arm 1 over arm 2; and $Q_1 - \widebar Q_1$ is the value of optimally allocating arm 1 based on covariates vs.~spending the same budget allocating arm 1 to a random subset of units, at a given budget level. Our proposed Qini curve extension to multiple arms facilitates policy value comparisons across all entries in Table \ref{tab:targetingsetup}, for example, $Q_{1,2,3} - Q_1$ measures the value of targeting with all available arms over targeting with only arm 1.

\begin{table}[ht]
\begin{center}
\begin{tabular}{l|ll}
      \multicolumn{3}{r}{Use covariates $X_i$}\\
      Arms & Yes & No \\
      \hline
      1 & $Q_{1}$ & $\widebar Q_{1}$ \\
      2 & $Q_{2}$ & $\widebar Q_{2}$ \\
      3 & $Q_{3}$ & $\widebar Q_{3}$ \\
      1,~2 & $Q_{1,2}$ & $\widebar Q_{1,2}$ \\
      1,~3 & $Q_{1,3}$ & $\widebar Q_{1,3}$ \\
      2,~3 & $Q_{2,3}$ & $\widebar Q_{2,3}$ \\
      1,~2,~3 & $Q_{1,2,3}$ & $\widebar Q_{1,2,3}$ \\
    \end{tabular}
\caption{Possible policy value configurations with three arms at a given budget $B$. For example, $Q_{1}$ is the value of allocating using only arm 1, and $Q_{1,2,3}$ is the value of allocating using all arms available. The difference $Q_{1,2,3} - Q_1$ quantifies the value of targeting with all arms over targeting with only arm 1. The rightmost column indicates the corresponding policy value that ignores covariates, as in Definition \ref{def:pibar}.} \label{tab:targetingsetup}
\end{center}
\end{table}

To conduct hypothesis tests for the value of different targeting strategies, we can employ the central limit theorem in Section \ref{sec:inference} to construct asymptotically valid confidence intervals for the difference in policy values. To assess the value of targeting with an optimal combination of all available arms over using only one, or a subset of arms, we can employ the following pairwise comparison:
\begin{exam}(Value of employing more treatment arms).
Let $Q(B)$ be the Qini curve for the policy $\pi_B$ using all available arms, and let $Q_k(B)$ be the Qini curve for the policy $\pi^k_B$ using only the $k$-th arm (or a subset of all available arms, as denoted by the subscripts in Table \ref{tab:targetingsetup}). For a given $B$, a pointwise $1-\alpha$ confidence interval for the difference $Q(B) - Q_k(B)$ is 
$$
\widehat Q(B) - \widehat Q_k(B) \pm z_{1-\alpha/2}\hat \sigma,
$$
where $z$ are the standard normal quantiles and $\hat \sigma$ estimates of the standard deviation $\sigma (Q(B) - Q_k(B))$. The accompanying software estimates $\hat \sigma$ via a paired bootstrap using Algorithm \ref{alg:bootstrap}.
\end{exam}
Figure \ref{fig:figure_targeting_b} provides a stylized example of what Qini curves could look like in the scenario where there, depending on budget, is a benefit to using an optimal combination of arms over just a single arm, when there are 3 treatment arms (plus a control) available. For example, at $B=2$ the difference $Q(B) - Q_1(B)$ is the vertical difference between the curve for all arms and the curve for arm 1 and indicates that optimally selecting among all available arms can yield an increase in gain of around 1.5 over only targeting with arm 1.
\begin{figure}[t]
    \centering
    \begin{subfigure}[b]{0.475\textwidth}  
        \centering 
        \includegraphics[width=\textwidth]{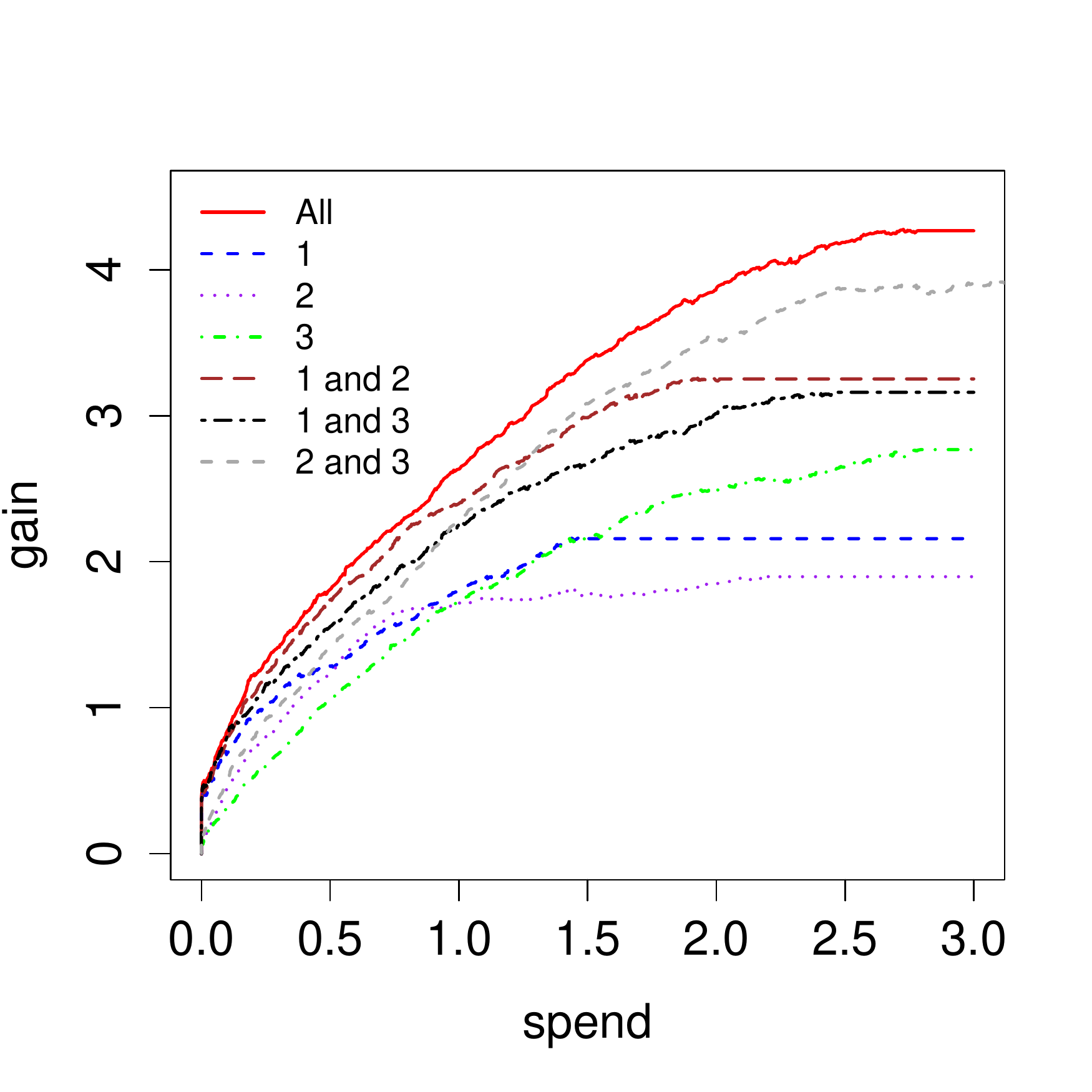}
        \caption[]
        {{The policy using all arms, and possible configurations of policies using the remaining arms.}}    
        \label{fig:figure_targeting_b}
    \end{subfigure}
    \hfill
    \begin{subfigure}[b]{0.475\textwidth}
        \centering
        \includegraphics[width=\textwidth]{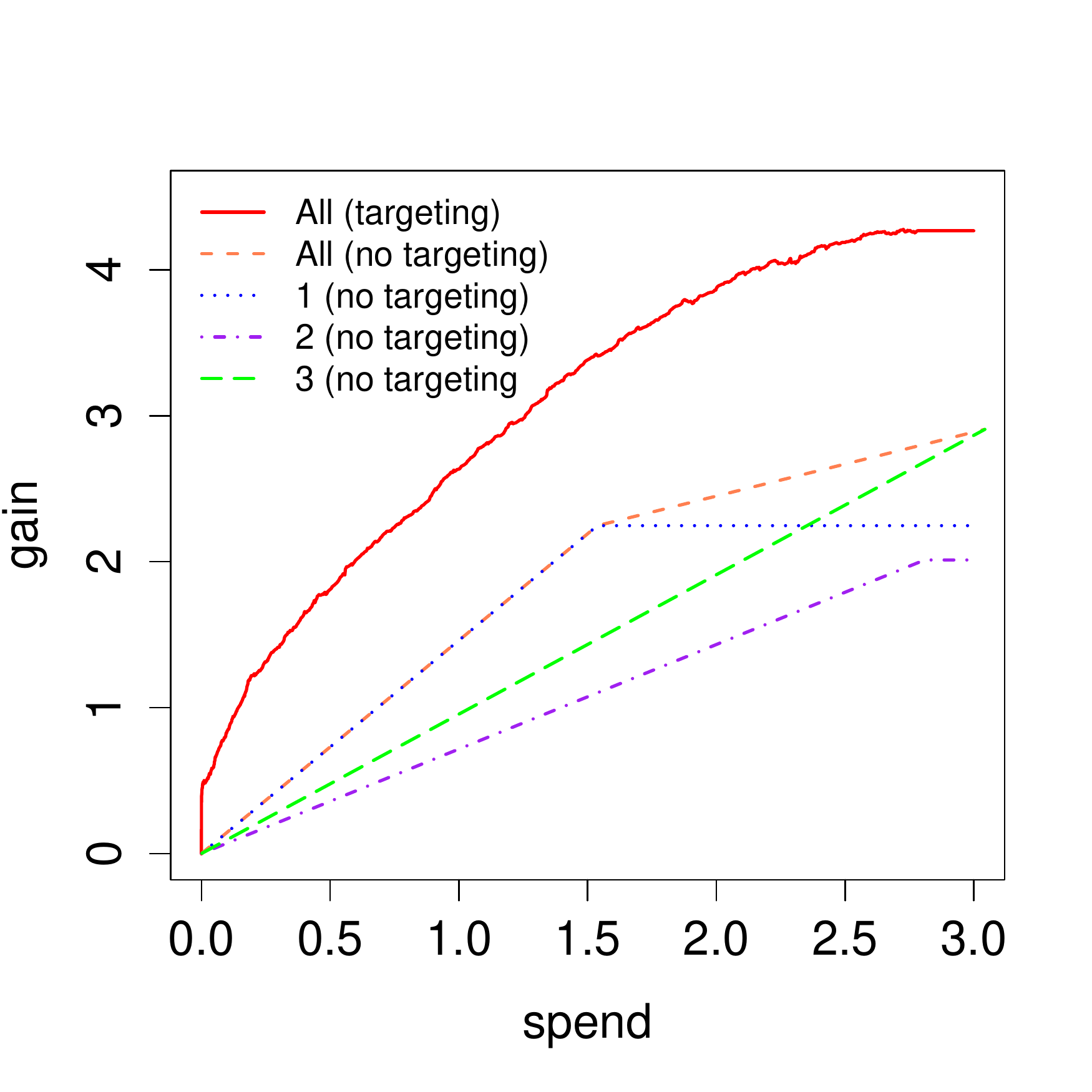}
        \caption[]
        {{The policy using all arms, with and without targeting, and three single-armed policies without targeting.}}    
        \label{fig:figure_targeting_a}
    \end{subfigure}    
    \caption{Stylized example of Qini curves with three treatment arms and a control. The first line segment of the multi-armed policy without targeting (b) traces out an arm-1 average, and the second segment traces out a convex combination of an arm-1 and arm-3 average.}
\end{figure}
To assess the value of targeting based on covariates, we can employ a similar pairwise comparison:
\begin{exam}(Value of targeting).
For a given $B$, a pointwise $1-\alpha$ confidence interval for the difference $Q(B) - \widebar Q(B)$ is
$$
\widehat Q(B) - \widehat{\widebar Q}(B) \pm z_{1-\alpha/2}\hat \sigma,
$$
where $z$ are the standard normal quantiles and $\hat \sigma$ an estimate of the standard deviation $\sigma (Q(B) - \widebar Q(B))$.
\end{exam}

Figure \ref{fig:figure_targeting_a} illustrates how Qini curves may look in the scenario where there is a benefit to targeting based on subject characteristics. For a fixed spend point and policy, the quantity $Q(B) - \widebar Q(B)$ measures the vertical difference between the targeting curve for all arms and one of the remaining lines, which signifies a baseline policy using all or only a single arm without targeting. Since this distance is positive, it signifies a benefit of targeting based on subject characteristics.

To verify the practical performance of the hypothesis test constructions in this Section, we revisit the simulation setup in Section \ref{sec:simulation} and repeat the same exercise as in Table \ref{tab:CI}, but for five different policy value comparisons with standard errors calculated via a paired bootstrap (that takes into account the correlation between curves evaluated on the same test set data). The results in Table \ref{tab:CI_paired} indicate these constructions can be justified in practice.

\begin{rema}
    The natural area under the curve counterparts for metrics in this section would be the integrated difference. For example, given some chosen maximum budget $\widebar B$ the quantity $\int_{0}^{\widebar B} \left(Q(B) - Q_k(B)\right) dB$ would estimate the area between two curves in Figure \ref{fig:figure_targeting_b}. We consider this an interesting extension but leave the development of such a functional central limit theorem to future work.
\end{rema}

\begin{table}[p]
\begin{center}
\subcaption*{Panel A: $Q_{1,2} - \widebar Q_{1,2}$}
\begin{tabular}{c|rrrrrrrrrr}
      \multicolumn{6}{r}{Spend ($B$)}\\
      Sample size & 0.05 & 0.1 & 0.15 & 0.2 & 0.25 & 0.3 & 0.35 & 0.4 & 0.45 & 0.5 \\
      \hline
1000 & 0.95 & 0.96 & 0.96 & 0.97 & 0.96 & 0.95 & 0.96 & 0.96 & 0.96 & 0.96 \\ 
  2000 & 0.95 & 0.95 & 0.95 & 0.94 & 0.95 & 0.95 & 0.95 & 0.95 & 0.95 & 0.95 \\ 
  5000 & 0.95 & 0.95 & 0.95 & 0.95 & 0.95 & 0.95 & 0.95 & 0.95 & 0.95 & 0.95 \\ 
  10000 & 0.95 & 0.94 & 0.94 & 0.94 & 0.95 & 0.95 & 0.95 & 0.95 & 0.95 & 0.95 \\ 
\end{tabular}

\bigskip
\subcaption*{Panel B: $Q_{1,2} - Q_1$}
\begin{tabular}{c|rrrrrrrrrr}
      \multicolumn{6}{r}{Spend ($B$)}\\
      Sample size & 0.05 & 0.1 & 0.15 & 0.2 & 0.25 & 0.3 & 0.35 & 0.4 & 0.45 & 0.5 \\
      \hline
1000 & 0.95 & 0.95 & 0.96 & 0.96 & 0.95 & 0.95 & 0.94 & 0.95 & 0.95 & 0.95 \\ 
  2000 & 0.95 & 0.95 & 0.95 & 0.93 & 0.95 & 0.95 & 0.94 & 0.94 & 0.95 & 0.95 \\ 
  5000 & 0.95 & 0.96 & 0.95 & 0.95 & 0.94 & 0.95 & 0.95 & 0.96 & 0.96 & 0.95 \\ 
  10000 & 0.95 & 0.94 & 0.94 & 0.94 & 0.95 & 0.95 & 0.96 & 0.96 & 0.95 & 0.95 \\ 
\end{tabular}

\bigskip
\subcaption*{Panel C: $Q_{1,2} - Q_2$}
\begin{tabular}{c|rrrrrrrrrr}
      \multicolumn{6}{r}{Spend ($B$)}\\
      Sample size & 0.05 & 0.1 & 0.15 & 0.2 & 0.25 & 0.3 & 0.35 & 0.4 & 0.45 & 0.5 \\
      \hline
1000 & 0.95 & 0.95 & 0.96 & 0.95 & 0.96 & 0.95 & 0.95 & 0.96 & 0.94 & 0.94 \\ 
  2000 & 0.96 & 0.95 & 0.97 & 0.96 & 0.96 & 0.96 & 0.96 & 0.95 & 0.95 & 0.96 \\ 
  5000 & 0.96 & 0.95 & 0.95 & 0.95 & 0.95 & 0.94 & 0.95 & 0.95 & 0.96 & 0.95 \\ 
  10000 & 0.95 & 0.95 & 0.96 & 0.96 & 0.95 & 0.95 & 0.96 & 0.96 & 0.95 & 0.95 \\ 
\end{tabular}

\bigskip
\subcaption*{Panel D: $Q_1 - \widebar Q_{1,2}$}
\begin{tabular}{c|rrrrrrrrrr}
      \multicolumn{6}{r}{Spend ($B$)}\\
      Sample size & 0.05 & 0.1 & 0.15 & 0.2 & 0.25 & 0.3 & 0.35 & 0.4 & 0.45 & 0.5 \\
      \hline
1000 & 0.96 & 0.95 & 0.95 & 0.96 & 0.95 & 0.96 & 0.96 & 0.96 & 0.96 & 0.95 \\ 
  2000 & 0.95 & 0.96 & 0.95 & 0.96 & 0.96 & 0.96 & 0.96 & 0.96 & 0.96 & 0.95 \\ 
  5000 & 0.95 & 0.95 & 0.95 & 0.94 & 0.94 & 0.95 & 0.95 & 0.95 & 0.96 & 0.95 \\ 
  10000 & 0.95 & 0.95 & 0.95 & 0.94 & 0.94 & 0.95 & 0.95 & 0.95 & 0.95 & 0.95 \\ 
\end{tabular}

\bigskip
\subcaption*{Panel E: $Q_2 - \widebar Q_{1,2}$}
\begin{tabular}{c|rrrrrrrrrr}
      \multicolumn{6}{r}{Spend ($B$)}\\
      Sample size & 0.05 & 0.1 & 0.15 & 0.2 & 0.25 & 0.3 & 0.35 & 0.4 & 0.45 & 0.5 \\
      \hline
1000 & 0.95 & 0.95 & 0.96 & 0.95 & 0.95 & 0.95 & 0.96 & 0.96 & 0.95 & 0.96 \\ 
  2000 & 0.95 & 0.96 & 0.95 & 0.94 & 0.95 & 0.95 & 0.95 & 0.95 & 0.95 & 0.95 \\ 
  5000 & 0.95 & 0.95 & 0.95 & 0.95 & 0.94 & 0.95 & 0.95 & 0.95 & 0.94 & 0.95 \\ 
  10000 & 0.94 & 0.95 & 0.94 & 0.95 & 0.96 & 0.95 & 0.96 & 0.94 & 0.95 & 0.95 \\ 
\end{tabular}
\end{center}
\caption{Coverage (\%) of the 95\% confidence intervals for paired differences at ten spend points using the simulation setup described in Section \ref{sec:simulation}. The number of Monte Carlo repetitions is 1000. The number of bootstrap replicates for standard error estimation is 200.}
\label{tab:CI_paired}
\end{table}

\afterpage{\FloatBarrier}

\section{Application: Treatment Targeting for Election Turnout}\label{sec:voting}
\citet{gerber2008social} conducts a multi-armed randomized controlled trial to study the social determinants of voter turnout in the 2006 US primary election, by mailing out letters of various forms. 180 002 households were randomly assigned one of $K=4$ treatment arms where arm 1 (``Civic'') tells the recipient to do their civic duty and vote. Arm 2 (``Hawthorne'') informs the recipient that their decision to vote or not is being monitored. Arm 3 (``Self'') informs the recipient about their and similar households' past voting history, and arm 4 (``Neighbors'') will let the recipient's neighbors know about their voting history. The control group receives no letter. The outcome of interest is whether a person in the household votes in the upcoming primary election. \citet{gerber2008social} finds that sending out the ``Neighbors'' letter is the most effective at increasing voter turnout, with little evidence of heterogeneity.

\begin{figure}[t]
    \centering
    \begin{subfigure}[b]{0.475\textwidth}
        \centering
        \includegraphics[width=\textwidth]{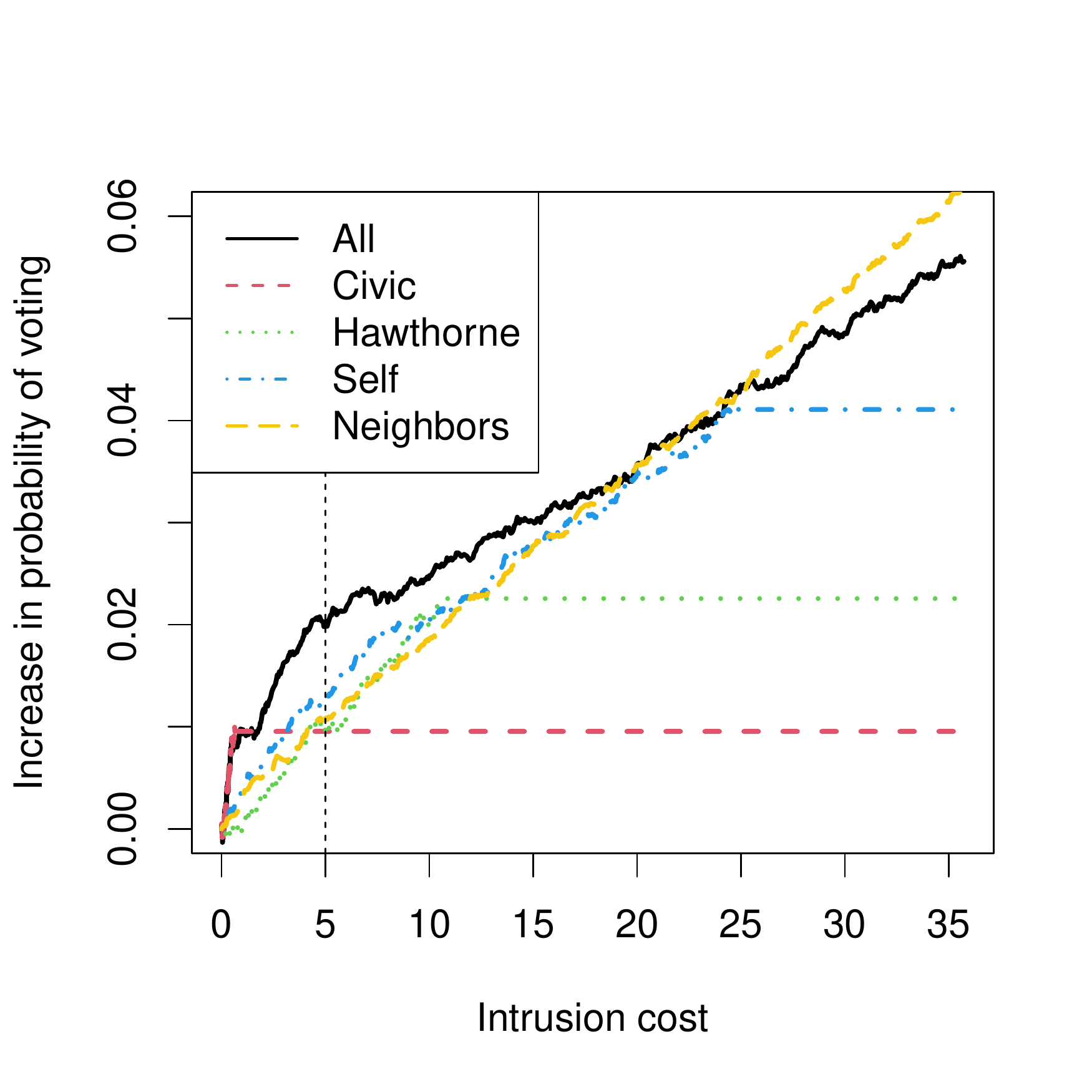}
        \caption[]
        {{Targeting with all arms vs. targeting with single arms.}}
        \label{fig:figure_voting_a}
    \end{subfigure}
    \hfill
    \begin{subfigure}[b]{0.475\textwidth}  
        \centering 
        \includegraphics[width=\textwidth]{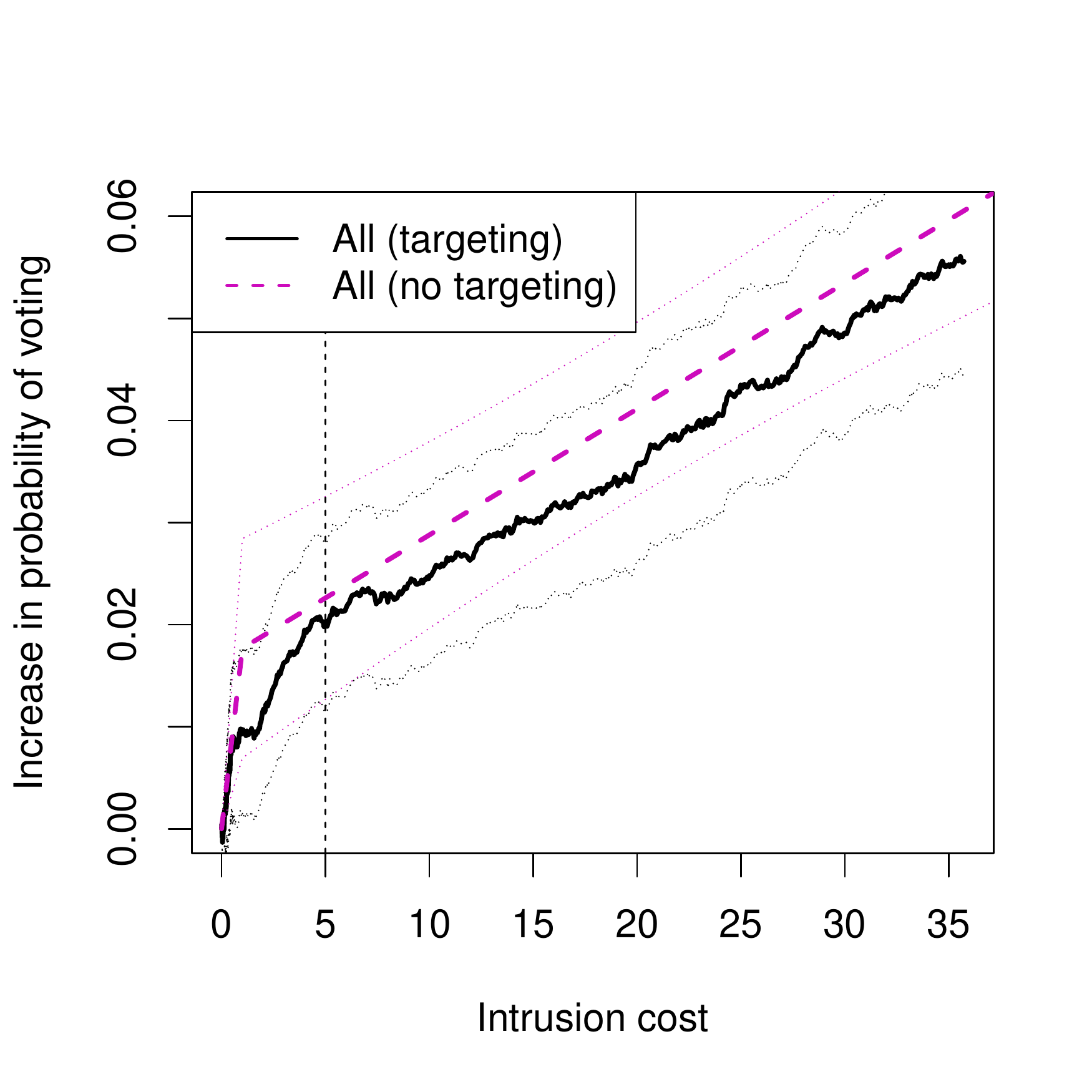}
        \caption[]
        {{All arms with targeting vs. all arms without targeting.}}  
        \label{fig:figure_voting_b}
    \end{subfigure}
    \caption{(a) Qini curves estimated on data from \citet{gerber2008social}, evaluated with inverse-propensity weighting using known randomization probabilities. (b) Qini curves for the multi-armed policy and the baseline multi-armed policy without targeting. The 95\% confidence intervals (dotted lines) shown in (b) are pointwise and are left out of (a) for legibility.}
\end{figure}

This treatment arm choice is intrusive, and to characterize the tradeoff between increases in voter turnout and incurred ``intrusions'', and to investigate whether targeting with less aggressive options might give similar increases in turnouts, we utilize Qini curves. The publicly available dataset from \citet{gerber2008social} includes variables that are associated with election turnout, and that we use to train a $\hat \tau(\cdot)$ function using the \texttt{grf} estimator described in Section \ref{sec:simulation}. These covariates include age, year of birth, gender, and household size, as well as six binary variables indicating if the subject voted in the general and primary elections in the years 2000, 2002, and 2004. To evaluate policies we hold out a random half-sample of the households and use inverse-propensity weighting with the known randomization probabilities $\PP{W_i = k} = 1/9~ (k=1\ldots4)$, where the control arm $k=0$ has assignment probability $5/9$. To incorporate costs, there are many modeling approaches one might take. In this example, we denominate costs in ``intrusion units'' where treatment arm 1 is least intrusive with $C_i(1) = 1$, then measure costs of the remaining arms as some multiple of this.

Here, we assume arm costs of $C_i(2)=15, C_i(3)=30, C_i(4)=45$, resulting in the Qini curves shown in Figure \ref{fig:figure_voting_a}. These curves suggest a benefit of multi-armed treatment rules over single-armed ones. For example, an optimal combination of the arms can at a budget level $B=5$ yield an increase in voter turnout of $2\%$ (95\% CI: $[1.2, 2.8]$) where only the least intrusive first arm  (``Civic'') would yield $1\%$ (95\% CI: $[0.1, 1.8]$). A paired test between the two policies at $B=5$ yields a 95\% CI of $[0.5, 1.5]$ for the difference in Qini curves.

Figure \ref{fig:figure_voting_b} compares the Qini curves from using all arms together with the non-targeting baseline policy $\bar \pi_B$, which in this case will only allocate between the non-intrusive arm 1 (``Civic) and the most intrusive but effective arm 4 (``Neighbors''). At $B=5$ a random 91\% of units are assigned arm 1 and the remaining 9\% arm 4 for a gain that is practically the same as for the targeting policy using all arms (95\% CI for the difference $Q_{1,2,3,4} - \widebar Q_{1,2,3,4}:~ [-0.9, 0.4]$).
Thus, in this application, it appears that the gains seen in Figure \ref{fig:figure_voting_a} are not primarily due to targeting---but rather to the flexibility offered by multi-armed policies in splitting budget between one arm that is cheaper but less effective and another that is more effective but also more expensive.

\section*{Acknowledgement}
We thank the Golub Capital Social Impact Lab at Stanford Graduate School of Business, and the Office of Naval Research (grants N00014-19-1-2468 and N00014-22-1-2668) for their financial support of this research. We are also grateful to Vitor Hadad and Emil Palikot for helpful feedback and to James Yang for helpful input on templatization in C++.


\newpage
\bibliographystyle{plainnat}
\bibliography{bibliography}

\clearpage
\newpage
\appendix

\section{Algorithm Details}
\subsection{Computing the Upper Left Convex Hull}\label{sec:appendix_cvx_hull}
The reduction to convex hulls in Algorithm \ref{alg:algopath} in the function \texttt{ComputeConvexHull} can be done using a variant of the Graham scan \citep{graham1972efficient}. Consider treatment arms $h, j, l$ sorted according to costs $C_h(X_i) < C_j(X_i) < C_l(X_i)$. To construct the hull, start with the two least costly arms $h$ and $j$ added to the hull, then do a linear scan through the remaining arms in order of increasing cost and determine if the $j$-th arm should be kept or removed from the hull by checking if the slope (as defined in Figure \ref{fig:lp_dominated}) from $j$ to $l$ is larger than the slope from $h$ to $j$. If the slope is larger, $j$ is removed, otherwise, $j$ is kept. If all elements of $\hat \tau(X_i)$ are negative, the convex hull for that unit is defined to be empty.

\subsection{Time Complexity of Algorithm \ref{alg:algopath}}
Given $n$ test samples, the run time of computing the multi-armed policy path is $O(nK \log{nK} + nK \log K) = O(nK \log{nK})$. To see this, note that the worst-case run time of Algorithm \ref{alg:algopath} occurs when for every unit each arm lies on the convex hull, and the maximum budget exceeds the expected cost of the most costly arm, i.e. $B_{max} > \EE{C_{k_{m_x}(X_i)}(X_i)}$. The convex hulls then have total size $nK$, and since the final budget constraint will never bind, a total of $nK$ items have to be inserted into the priority queue, which takes time $O(nK \log{nK})$. Computing the convex hull involves sorting each unit's cost in increasing order, which takes time $O(K \log K)$, and this has to be repeated $n$ times, yielding the claimed run time.

\section{Proofs}\label{sec:Proofs}
\subsection{Proof of Theorem \ref{theo:opt_policy}}
\begin{proof}
 Assume $X$ is a random draw from the covariate distribution and $X_i$ are i.i.d. We first note that in our multi-armed case, the policy $\pi$ is a vector and the expected cost can be written as
 \begin{equation*}
     \Psi(\pi) = \EE{\sum_{k=1}^K \pi_k(X_i) C_k(X_i)}. 
 \end{equation*}
 Consider the following function of $\lambda$, 
 \begin{equation*}
     \beta(\lambda) = \EE{\sum_{j=1}^{m_x} \mathbf{1}\left(\rho_{k_j(x)}(x) > \lambda > \rho_{k_{j+1}(x)}(x)\right)C_{k_j(x)}(x)}.
 \end{equation*}
 By our assumption, we see it is a non-increasing function of $\lambda$. Let 
 \begin{equation}
     \eta_B := \inf \{\lambda: \beta(\lambda) \leq B\}, \, \lambda_B = \max \{\eta_B, 0\}.
 \end{equation}
 Then the policy \eqref{eq:pi_opt} could be rewritten as 
 \begin{equation} \label{eq:pi_opt_rewritten}
     \pi_{B, k_j(x)}^*(x) = 
\begin{cases}
    c &\text{if} \ \rho_{k_j(x)} = \lambda_B, \\
    1-c &\text{if} \ \rho_{k_{j-1}(x)} = \lambda_B, \\
    1 & \text{if} \ \rho_{k_j(x)}(x) > \lambda_B > \rho_{k_{j+1}(x)}(x)
\end{cases}
\end{equation}
where
\begin{equation} \label{eq:c_B}
     c = 
\begin{cases}
    0 &\text{if} \ \EE{\sum_{j=1}^{m_x} \mathbf{1}(\rho_{k_j(x)}(x) = \lambda_B )C_{k_j(x)}(x)} = 0, \\
    \frac{B-\EE{\sum_{j=1}^{m_x} \mathbf{1}\left(\rho_{k_j(x)}(x) > \lambda > \rho_{k_{j+1}(x)}(x)\right)C_{k_j(x)}(x)}}{\sum_{j=1}^{m_x} \mathbf{1}(\rho_{k_j(x)}(x) = \lambda_B )C_{k_j(x)}(x)} & \text{if} \ \EE{\sum_{j=1}^{m_x} \mathbf{1}(\rho_{k_j(x)}(x) = \lambda_B )C_{k_j(x)}(x)} > 0
\end{cases}
\end{equation}
Now we prove the above rule is in fact optimal. Let $\pi'(x)$ denote any other stochastic treatment rule that satisfies the budget constraint. We want to argue 
\begin{equation*}
\EE{\sum_{k=1}^{K}\pi_k(X)\tau_k(X)} \geq \EE{\sum_{k=1}^{K} \pi'_k(X)\tau_k(X)}.
\end{equation*}
To prove this, we have 
\begin{align}
   & \EE{\sum_{k=1}^{K}\pi_k(X)\tau_k(X)} - \EE{\sum_{k=1}^{K} \pi'_k(X)\tau_k(X)} \notag \\
   &= \EE{\EE{\sum_{k=1}^{K} (\pi_k(X)-\pi_k'(X))\tau_k(X) \mid X}} \notag \\
   &= \int \EE{\sum_{j=1}^{K} (\pi_{k_j(x)}(x) -\pi'_{k_j(x)}(x) )\tau_{k_j(x)}(x)} dP(x), \label{eq:proof_goal}
\end{align}
where we define $k_{m_x+1}(x),\ldots,k_K(x)$ to be any ordering of points not in the convex hull. 
Now we have 
\begin{align*}
    &\sum_{j=1}^{K} (\pi_{k_j(x)}(x) -\pi'_{k_j(x)}(x) )\tau_{k_j(x)}(x) \\
    &= \sum_{j=1}^{K}(\pi_{k_j(x)}(x) -\pi'_{k_j(x)}(x) ) \sum_{l=1}^{j}(\tau_{k_l(x)}(x)-\tau_{k_{l-1}(x)}(x)) \\
    & = \sum_{l=1}^{K} (\tau_{k_l(x)}(x)-\tau_{k_{l-1}(x)}(x)) \sum_{j=l}^{K}(\pi_{k_j(x)}(x) -\pi'_{k_j(x)}(x)) \\
    & = \sum_{l=1}^{K} (C_{k_l(x)}(x)-C_{k_{l-1}(x)}(x)) \frac{\tau_{k_l(x)}(x)-\tau_{k_{l-1}(x)}(x)}{C_{k_l(x)}(x)-C_{k_{l-1}(x)}(x)}\sum_{j=l}^{K}(\pi_{k_j(x)}(x) -\pi'_{k_j(x)}(x)) \\
    &= \sum_{l=1}^{K} \rho_{k_l(x)}(x)(C_{k_l(x)}(x)-C_{k_{l-1}(x)}(x)) \sum_{j=l}^{K}(\pi_{k_j(x)}(x) -\pi'_{k_j(x)}(x)),
\end{align*}
where we use the fact that we assume $\tau_0(x) = 0$. Note that by our characterization of the optimal policy, there exists $k \in \{1,\ldots,K\}$ such that, $\rho_{k_l(x)}(x) \geq \lambda_B$ if $l \leq k$. In these cases by the definition of our policy $\pi$, we either have $\sum_{j=l}^{K} \pi_{k_j(x)}(x) = 1 \geq \sum_{j=l}^{K} \pi'_{k_j(x)}(x)$ or $\rho_{k_l(x)}(x) = \lambda_B$. If $l > k$ then $\sum_{j=l}^{K} \pi_{k_j(x)}(x) = 0 \leq \sum_{j=l}^{K} \pi'_{k_j(x)}(x)$ and $\rho_{k_l(x)}(x) < \lambda_B$. Combining the two cases we see 
\begin{align*}
    &\sum_{j=1}^{K} (\pi_{k_j(x)}(x) -\pi'_{k_j(x)}(x) )\tau_{k_j(x)}(x) \\
    & \geq \sum_{l=1}^{K} \lambda_B(C_{k_l(x)}(x)-C_{k_{l-1}(x)}(x)) \sum_{j=l}^{K}(\pi_{k_j(x)}(x) -\pi'_{k_j(x)}(x)) \\
    & = \lambda_B \sum_{j=1}^{K} C_{k_j(x)} (\pi_{k_j(x)}(x) - \pi'_{k_j(x)}(x)).
\end{align*}
Hence, we have 
\begin{equation}
    \EE{\sum_{k=1}^{K}\pi_k(X)\tau_k(X)} - \EE{\sum_{k=1}^{K} \pi'_k(X)\tau_k(X)} \geq \lambda_B \EE{\sum_{k=1}^{K}(\pi_{k}(X) -\pi'_{k}(X))C_{k}(X)}. \label{eq:last_step}
\end{equation}
Now we consider two cases: Either $\lambda_B > 0$ or $\lambda_B = 0$. If $\lambda_B > 0$, we have consumed all budget then obviously $\eqref{eq:last_step} \geq 0$ and if $\lambda_B = 0$, then we are done as well.  
\end{proof}

\subsection{Proof of Theorem \ref{theo:asymp_linear_reward}}
\begin{proof}
We proceed with three steps. First we argue that $\lamh_B$ is consistent, i.e. $\lamh_B \convp \lambda_B$. Second, we argue that $n^{1/2}(\lamh_B - \lambda_B)$ is asymptotically linear. Finally, we argue that 
\[
n^{1/2}\left(\widehat V(T(\cdot \, ; \, \hat \rho, \, \hat \lambda_B, \, \hat c_B))- V(T(\cdot \, ; \, \hat \rho, \, \lambda_B, \, c_B))\right)
\]is asymptotic linear. \\\\
\textit{Step 1}: $\lamh_B \convp \lambda_B$.\\
    We use Theorem 5.9 of \cite{vaart_1998}. We need to verify the uniform convergence of 
    \begin{equation}
    \Psi_n(T(\cdot \, ; \, \hat \rho, \, \lambda, \, c)) - B = n^{-1} \sum_{i=1}^{n} \langle T(X_i; \, \hat \rho, \, \lambda, \, c),  C(X_{i}) \rangle - B
    \end{equation} to $\Psi(T(\cdot \, ; \, \hat \rho, \, \lambda, \, c)) - B$. We first prove a lemma.  
    \begin{lemm} \label{lemm:donsker}
Suppose $g_1$, $g_2$ and $h$ are measurable functions from $\mathbb{R}^d$ to $\mathbb{R}$ such that for any $x, \,h(x) \leq M$ and $g_1(x) > g_2(x)$, then the function class $\{f_{\lambda}(x):= \mathbf{1}(g_1(x) > \lambda > g_2(x))h(x), \lambda \in [0, L]\}$ is $P$-Donsker for any law $P$ on $\mathcal{X}$. 
\end{lemm}
\begin{proof}
    We note that 
    \begin{equation}
        f_{\lambda}(x) = (\mathbf{1}(g_1(x) > \lambda) - \mathbf{1}(g_2(x) > \lambda))h(x).
    \end{equation}
The indicator functions are a VC class hence Donsker and $h(x)$ is uniformly bounded. Hence $f_{\lambda}$ is also Donsker. 
\end{proof}
By Lemma \ref{lemm:donsker} and the fact that the finite sum of a Donsker class is also Donsker, we know $\langle T(X_i; \, \hat \rho, \, \lambda, \, c),  C(X_{i}) \rangle - B$ indexed by $\lambda$ forms a Donsker class. In particular, it is Glivenko-Cantelli and the uniform convergence holds. Now we verify the second condition in the theorem. By our assumption, $\Psi(T(\cdot \, ; \, \hat \rho, \, \lambda, \, c)) - B$ is continuously differentiable, and also by our definition and assumptions, we know $\Psi(T(\cdot \, ; \, \hat \rho, \, \lambda, \, c)) - B$ is monotonically decreasing. In particular, it has a well-defined inverse. This verifies the second condition in the theorem. Finally, our $\lamh_B$ solves the estimating equation approximately, and by Theorem 5.9 of \cite{vaart_1998}, $\lamh_B$ is consistent. \\ \\
\textit{Step 2}: $n^{1/2}(\lamh_B - \lambda_B)$ is asymptotic linear. \\
We use Theorem 5.21 of \cite{vaart_1998}. To verify the convergence (5.22) in the proof, we use Lemma 19.24 and the following additional lemma. 
\begin{lemm} \label{lemm:l2_conv}
    Suppose $\lamh \convp \lambda$, and $f_{\lambda}$ is defined as in Lemma \ref{lemm:donsker}, then $\|f_{\lamh} - f_{\lambda}\|_2^2 \convp 0$.
\end{lemm}
\begin{proof}
   We note by dominated convergence theorem, if the sequence $\lambda_n \rightarrow \lambda$ almost surely, then $\|f_{\lambda_n} - f_{\lambda}\|_2^2 \rightarrow 0$ almost surely. Now fix a subsequence $n_k$, since $\hat{\lambda}_{n_k} \convp \lambda$, we know there is a further subsequence $n(m_k)$ such that $\hat{\lambda}_{n(m_k)} \rightarrow \lambda$ almost surely. Then by the above argument, $\|f_{\lambda_{n(m_k)}} - f_{\lambda}\|_2^2 \rightarrow 0$ almost surely, which establishes the convergence in probability since every subsequence has a further subsequence that converges almost surely.
\end{proof}
Since the function $\psi_{\lambda}(X_i) = \langle T(X_i; \, \hat \rho, \, \lambda, \, c),  C(X_{i}) \rangle - B$ when viewed as indexed by $\lambda$ is a finite sum of functions of the form $f_{\lambda}$, the above lemma also holds for $\psi_{\lambda}$. By Lemma 19.24 of \cite{vaart_1998}, we know 
\begin{equation}
    \mathbb{G}_n\psi_{\lamh_B} - \mathbb{G}_n\psi_{\lambda_B} \convp 0.
\end{equation}
To apply Theorem 5.21 we also need to show that the derivative of $\Psi(T(\cdot \, ; \, \hat \rho, \, \lambda, \, c))$ is nonzero at $\lambda_B$. To prove this, for simplicity, we assume there is only one treatment arm in addition to the control arm, then we have 
\begin{equation}
     \Psi'(T(\cdot \, ; \, \hat \rho, \, \lambda, \, c)) = p(\lambda)\EE{C_i(1) - C_i(0) \mid \hat \rho_i = \lambda},
\end{equation}
where $p$ is the density function of the incremental ratio $\hat \rho_i$. By our assumption on the density of $\hat \rho$, we know this is greater than zero. Finally by our boundedness assumption, $\psi_{\lambda}$ is $L_2$, by \eqref{eq:approxlambda}, $\lamh_B$ approximately solves the estimating equation by $o_p(n^{-1/2})$ and $\lamh_B$ is consistent by step 1. By Theorem 5.21 of \cite{vaart_1998}, we have 
\begin{equation}
    n^{1/2}(\lamh_B-\lambda_B) = -\frac{1}{\Psi'(\pi_B)}n^{-1/2}\sum_{i=1}^{n}\psi_{\lambda_B}(X_i) + o_p(1).
\end{equation}
\textit{Step 3}: $n^{1/2}\left(\widehat V(T(\cdot \, ; \, \hat \rho, \, \hat \lambda_B, \, \hat c_B))- V(T(\cdot \, ; \, \hat \rho, \, \lambda_B, \, c_B))\right)$ is asymptotic linear.\\
Define $\tV(T(\cdot \, ; \, \hat \rho, \, \lambda, \, c)) = n^{-1}\sum_{i=1}^{n} \langle T(X_i; \, \hat \rho, \, \lambda \, ,c), \tau(X_i)\rangle$ and recall \\$V(T(\cdot \, ; \, \hat \rho, \, \lambda, \, c)) = \EE{\langle T(X_i; \, \hat \rho, \, \lambda, \, c), \tau(X_i) \rangle}$. We have the following decomposition
\begin{align}
    & n^{1/2}\left(\widehat V(T(\cdot \, ; \, \hat \rho, \, \hat \lambda_B, \, \hat c_B))- V(T(\cdot \, ; \, \hat \rho, \, \lambda_B, \, c_B))\right) \notag \\
    & =  n^{1/2}\left(\widehat V(T(\cdot \, ; \, \hat \rho, \, \hat \lambda_B, \, \hat c_B)) - \tV(T(\cdot \, ; \, \hat \rho, \, \hat \lambda_B, \, \hat c_B))\right) \label{term1} \\
    & + n^{1/2}\left(\tV(T(\cdot \, ; \, \hat \rho, \, \hat \lambda_B, \, \hat c_B)) - V(T(\cdot \, ; \, \hat \rho, \, \hat \lambda_B, \, \hat c_B))\right) \label{term2}\\
    & + n^{1/2}\left(V(T(\cdot \, ; \, \hat \rho, \, \hat \lambda_B, \, \hat c_B)) - V(T(\cdot \, ; \, \hat \rho, \, \lambda_B, \, c_B))\right). \label{term3}
\end{align}
We will deal with the three terms one by one. For \eqref{term1}, we have 
\begin{align}
    & n^{1/2}\left(\widehat V(T(\cdot \, ; \, \hat \rho, \, \hat \lambda_B, \, \hat c_B)) - \tV(T(\cdot \, ; \, \hat \rho, \, \hat \lambda_B, \, \hat c_B))\right) \notag \\
    &= n^{-1/2}\sum_{i=1}^n \langle T(X_i; \, \hat \rho, \, \hat \lambda_B, \, \hat c_B), \widehat \Gamma_i - \tau(X_i)\rangle \\
    & = n^{-1/2}\sum_{i=1}^n \langle T(X_i; \, \hat \rho, \, \hat \lambda_B, \, \hat c_B), \Gamma_i - \tau(X_i)\rangle + o_p(1) \label{eq:DR}\\
    & = n^{-1/2}\sum_{i=1}^n \langle T(X_i; \, \hat \rho, \, \lambda_B, \, c_B), \Gamma_i - \tau(X_i)\rangle + o_p(1), \label{eq:cmt}
\end{align}
where \eqref{eq:DR} follows from the boundedness of $T(\cdot \, ; \, \hat \rho, \, \lambda, \, c)$ and the usual analysis on doubly robust scores which gives for any $k$, 
\begin{equation}
    n^{-1/2} \sum_{i=1}^{n} (\widehat \Gamma_{i,k} - \Gamma_{i,k}) = o_p(1).
\end{equation}
To get \eqref{eq:cmt}, we note that we only need to prove 
\begin{equation} \label{goal:diff}
    n^{-1/2}\sum_{i=1}^n \langle T(X_i; \, \hat \rho, \, \hat \lambda_B, \, \hat c_B) - T(X_i; \, \hat \rho, \, \lambda_B, \, c_B), \Gamma_i - \tau(X_i)\rangle = o_p(1).
\end{equation}
To this end, note that \eqref{goal:diff} is mean 0 and we argue that the variance goes to 0. 
\begin{align*}
    & \EE{\left(n^{-1/2}\sum_{i=1}^n \langle T(X_i; \, \hat \rho, \, \hat \lambda_B, \, \hat c_B) - T(X_i; \, \hat \rho, \, \lambda_B, \, c_B), \Gamma_i - \tau(X_i)\rangle\right)^2} \\
    & = \EE{\EE{\left(n^{-1/2}\sum_{i=1}^n \langle T(X_i; \, \hat \rho, \, \hat \lambda_B, \, \hat c_B) - T(X_i; \, \hat \rho, \, \lambda_B, \, c_B), \Gamma_i - \tau(X_i)\rangle\right)^2 \mid X_{train}, X_{test}}} \\
    & = \EE{n^{-1}\sum_{i=1}^{n} \EE{\langle T(X_i; \, \hat \rho, \, \hat \lambda_B, \, \hat c_B) - T(X_i; \, \hat \rho, \, \lambda_B, \, c_B), \Gamma_i - \tau(X_i)\rangle^2 \mid X_{train}, X_{test}}} \\
    & \leq M \EE{n^{-1}\sum_{i=1}^n\|T(X_i; \, \hat \rho, \, \hat \lambda_B, \, \hat c_B) - T(X_i; \, \hat \rho, \, \lambda_B, \, c_B)\|_2^2} \\
    &= M \EE{\|T(X_i; \, \hat \rho, \, \hat \lambda_B, \, \hat c_B) - T(X_i; \, \hat \rho, \, \lambda_B, \, c_B)\|_2^2},
\end{align*}
where $M$ is a bound on the second moment of $\Gamma_i$. By the dominated convergence theorem or lemma \ref{lemm:l2_conv}, we know 
\begin{equation*}
    \EE{\|T(X_i; \, \hat \rho, \, \hat \lambda_B, \, \hat c_B) - T(X_i; \, \hat \rho, \, \lambda_B, \, c_B)\|_2^2} = o(1).
\end{equation*}
For \eqref{term2}, we can use the same machinery (Lemma 19.24 of \cite{vaart_1998}) as in step 2 to argue that 
\begin{align}
    & n^{1/2}\left(\tV(T(\cdot \, ; \, \hat \rho, \, \hat \lambda_B, \, \hat c_B)) - V(T(\cdot \, ; \, \hat \rho, \, \hat \lambda_B, \, \hat c_B))\right) \notag \\
    & = n^{1/2}\left(\tV(T(\cdot \, ; \, \hat \rho, \, \lambda_B, \, c_B)) - V(T(\cdot \, ; \, \hat \rho, \, \lambda_B, \, c_B))\right) + o_p(1).
\end{align}
Finally for \eqref{term3}, by differentiability of $V$ and step 2, we can use the delta method to get 
\begin{align}
    & n^{1/2}\left(V(T(\cdot \, ; \, \hat \rho, \, \hat \lambda_B, \, \hat c_B)) - V(T(\cdot \, ; \, \hat \rho, \, \lambda_B, \, c_B))\right) \notag \\
    & = n^{1/2}V'(T(\cdot \, ; \, \hat \rho, \, \lambda_B, \, c_B))(\hat \lambda_B - \lambda_B) + o_p(1). 
\end{align}
Now combine all three terms, and we have the following 
\begin{align}
    & n^{1/2}\left(\widehat V(T(\cdot \, ; \, \hat \rho, \, \hat \lambda_B, \, \hat c_B))- V(T(\cdot \, ; \, \hat \rho, \, \lambda_B, \, c_B))\right) \notag \\
    &= n^{-1/2}\sum_{i=1}^n \langle T(X_i; \, \hat \rho, \, \lambda_B, \, c_B), \Gamma_i - \tau(X_i)\rangle \\ 
    & + n^{-1/2}\sum_{i=1}^n (\langle T(X_i; \, \hat \rho, \, \lambda_B, \, c_B), \tau(X_i)\rangle - V(T(\cdot \, ; \, \hat \rho, \, \lambda_B, \, c_B)) \\
    & -\frac{V'(T(\cdot \, ; \, \hat \rho, \, \lambda_B, \, c_B))}{\Psi'(T(\cdot \, ; \, \hat \rho, \, \lambda_B, \, c_B))}n^{-1/2}\sum_{i=1}^{n}\psi_{\lambda}(X_i) + o_p(1) \\
    & = n^{-1/2} \sum_{i=1}^n \left(\langle T(X_i; \, \hat \rho, \, \lambda_B, \, c_B), \Gamma_i\rangle -  \frac{V'(T(\cdot \, ; \, \hat \rho, \, \lambda_B, \, c_B))}{\Psi'(T(\cdot \, ; \, \hat \rho, \, \lambda_B, \, c_B))}\psi_{\lambda}(X_i) - V(T(\cdot \, ; \, \hat \rho, \, \lambda_B, \, c_B)) \right)+ o_p(1)
\end{align}
\end{proof}

\section{Coverage Simulation with Boosting-based AIPW}\label{appendix:sim}
\setcounter{table}{0}
\renewcommand\thetable{\thesection\arabic{table}}

\begin{table}[ht]
\begin{center}
\begin{tabular}{c|rrrrrrrrrr}
      \multicolumn{6}{r}{Spend ($B$)}\\
      Sample size & 0.05 & 0.1 & 0.15 & 0.2 & 0.25 & 0.3 & 0.35 & 0.4 & 0.45 & 0.5 \\
      \hline
1000 & 0.93 & 0.91 & 0.91 & 0.91 & 0.90 & 0.90 & 0.88 & 0.88 & 0.88 & 0.87 \\ 
  2000 & 0.93 & 0.94 & 0.92 & 0.92 & 0.91 & 0.90 & 0.90 & 0.90 & 0.89 & 0.88 \\ 
  5000 & 0.93 & 0.92 & 0.92 & 0.91 & 0.92 & 0.91 & 0.90 & 0.90 & 0.89 & 0.90 \\ 
  10000 & 0.92 & 0.94 & 0.93 & 0.92 & 0.93 & 0.92 & 0.91 & 0.91 & 0.91 & 0.90 \\ 
    \end{tabular}
\caption{Coverage (\%) of the 95\% confidence intervals for $Q(B)$ at ten spend points ($B$) using the simulation setup described in Section \ref{sec:simulation}, with AIPW nuisance components \eqref{eq:drAIPW} estimated using \texttt{xgboost} \citep{xgboostR} (with tuning parameters selected via simple 3-fold cross validation over a randomly drawn grid). The number of Monte Carlo repetitions is 1000. The number of bootstrap replicates for standard error estimation is 200.}
\label{tab:CI_boosting}
\end{center}
\end{table}

\end{document}